\documentclass[letterpaper]{article} 
\usepackage{aaai25}  
\usepackage{times}  
\usepackage{helvet}  
\usepackage{courier}  
\usepackage[hyphens]{url}  
\usepackage{graphicx} 
\usepackage{natbib}  
\usepackage{caption} 
\usepackage{chrkroer}

\usepackage{algorithm}
\usepackage{algorithmic}
\usepackage{amsmath}
\usepackage{amssymb}
\usepackage{mathtools}
\usepackage{natbib}

\usepackage{subcaption}
\usepackage{array}

\usepackage{mathtools}
\usepackage{natbib}

\usepackage{tikz}
\usetikzlibrary{arrows,shapes,automata,calc,matrix,backgrounds,petri, positioning}
\allowdisplaybreaks
\newcommand{\real}{\mathbb{R}}

\usepackage{newfloat}
\usepackage{listings}
\DeclareCaptionStyle{ruled}{labelfont=normalfont,labelsep=colon,strut=off} 
\floatstyle{ruled}
\newfloat{listing}{tb}{lst}{}
\floatname{listing}{Listing}
\title{Commitment to Sparse Strategies in Two-Player Games}
\author {
    Salam Afiouni\textsuperscript{\rm 1},
    Jakub \v{C}ern\'{y}\textsuperscript{\rm 1},
    Chun Kai Ling\textsuperscript{\rm 2},
    Christian Kroer\textsuperscript{\rm 1}
}
\affiliations {
    \textsuperscript{\rm 1}Columbia University, USA\\
    \textsuperscript{\rm 2}National University of Singapore, Singapore\\
    sa4316@columbia.edu, chunkail@nus.edu.sg, 
    jakub.cerny@columbia.edu, christian.kroer@columbia.edu
}

\begin{document}

\maketitle

\begin{abstract}
While Nash equilibria are guaranteed to exist, they may exhibit dense support, making them difficult to understand and execute in some applications. In this paper, we study $k$-sparse commitments in games where one player is restricted to mixed strategies with support size at most $k$. Finding $k$-sparse commitments is known to be computationally hard. We start by showing several structural properties of $k$-sparse solutions, including that the optimal support may vary dramatically as $k$ increases. These results suggest that naive greedy or double-oracle-based approaches are unlikely to yield practical algorithms. We then develop a simple approach based on mixed integer linear programs (MILPs) for zero-sum games, general-sum Stackelberg games, and various forms of structured sparsity. We also propose practical algorithms for cases where one or both players have large (i.e., practically innumerable) action sets, utilizing a combination of MILPs and incremental strategy generation. We evaluate our methods on synthetic and real-world scenarios based on security applications. In both settings, we observe that even for small support sizes, we can obtain more than $90\%$ of the true Nash value while maintaining a reasonable runtime, demonstrating the significance of our formulation and algorithms.
\end{abstract}

\section{Introduction}
Equilibrium finding in games has emerged as a key component of artificial intelligence, owing to applications in recreational game playing such as poker, Stratego, and Diplomacy, numerous applications in security, as well as logistics~\citep{bowling2015heads,brown2018superhuman,perolat2022mastering,meta2022human,pita2008deployed,fang2015security,tambe2011security,jain2010software,cerny2024contested}. Central to the success of these applications are efficient approaches towards finding the optimal \textit{randomized}, or \textit{mixed} equilibria in their respective settings. Not only are randomized equilibria typically guaranteed to exist, they allow for the prescription of behavior that hedges over all possible opponent actions, making them particularly useful in zero-sum games. Nonetheless, a persistent difficulty faced by practitioners seeking to apply game theoretic solutions into the real world is that such solutions exhibit \textit{dense support}, i.e., they could randomize over a large number of actions. Strategies with dense support are undesirable for many reasons; for instance, they are difficult to interpret and implement in practice, and are often faced with numerical issues if used for downstream processing. 

In this paper, we propose overcoming these limitations by computing $k$-sparse solutions as opposed to Nash equilibria. Loosely speaking, a distinguished player termed the sparse player (representing the practitioner) is restricted to playing a mixed strategy with a support size at most $k$. The other player is free to play any mixed strategy they desire. We argue that in many practical use-cases, sparse strategies can almost match the performance of dense strategies, as long as the support of the sparse player is carefully chosen.

The problem of finding sparse solutions has been studied from various angles. \citet{althofer1994sparse} show that sparse approximations to strategies exist in zero-sum games without too much degradation, while seminal work by \citet{lipton2003playing} show that \textit{approximate} sparse Nash equilibria always exist via a sampling argument (\citet{feder2007approximating} show the optimality of these results); follow up work extends this argument to correlated and other equilibria~\citep{babichenko2014simple}.
Nonetheless, these arguments are nonconstructive in nature, and finding ``good'' sparse equilibria can be difficult. Indeed, finding the \textit{optimal} $k$-sparse commitment even for zero-sum games is in general intractable \cite{mccarthy2018price}. 
\citet{foster2023hardness} study a variant of sparsity in Markov games where correlated equilibria distributions are expressible as a mixture of $k$ uniform product distributions, while \citet{anagnostides2023complexity} show hardness of computing sparse correlated equilibria for extensive form games.
The line of work closest to ours is by \citet{mccarthy2018price}, who study $k$-sparse solutions under the name of \textit{operationalizable} strategies, mainly from the perspective of security games. Our work differs from them in several ways. First, our algorithm is not restricted to finding $k$-uniform strategies, i.e., each of the $k$ actions (possibly repeated) must be chosen with equal probability. Crucially, we significantly expand the scope of applications beyond that of security games, including large games with MILP representable strategy spaces.

For the majority of theoretical and algorithmic work above, the term ``sparse'' is formally defined by asymptotics in terms of the size of the game (e.g., the number of players, and the size of their actions). 
For our contributions, we adopt colloquial, but often times more practically useful, meanings for the terms sparse and dense, with sparse meaning ``small enough'' (e.g., say, in single digits), while dense implies the converse. 
Secondly, a significant amount of work \cite{lipton2003playing,babichenko2014simple,mccarthy2018price} also requires sparse strategies to be $k$-uniform. This is a stronger restriction than our notion of $k$-sparsity, and it can adversely affect performance when we want a sparse strategy in the practical non-asymptotic sense where $k$ is a very small constant; we show this both empirically and theoretically in Sections~\ref{sec:theory} and \ref{sec:exp}.

With these distinctions in mind, our key contributions are as follows. First, we formalize the notion of $k$-sparse commitments for two-player zero-sum games and Stackelberg equilibrium in general-sum games. Second, we study some of the properties of $k$-sparse commitments. We show that even for zero-sum games, supports of the optimal $k$-sparse equilibrium can be disjoint for almost all $k$, and the value function as a set function of allowable support is not submodular. We also demonstrate that $k$-\textit{uniform} equilibria as computed by \citet{mccarthy2018price} can yield a much worse performance than $k$-sparse equilibrium. Thirdly, we show how optimal $k$-sparse commitments can be computed via Mixed Integer Linear Programs (MILPs) for a variety of settings, including zero-sum games, Stackelberg equilibrium, variants of structured action sets, as well as games where one or both players have a large, but MILP-representable strategy spaces. Finally, we demonstrate empirically the efficacy of our proposed method, using both randomly generated normal-form games, as well as settings based on security applications. Our results demonstrate scalability of our method which also often performs better than $k$-uniform strategies. Even for small $k$, we are often able to capture nearly $90\%$ of the true optimal commitment, showing that $k$-sparse commitments are a practical and viable.
\section{Preliminaries}
We are concerned with two-player games, where Player 1 and 2 have $n$ and $m$ actions respectively. We denote their payoff matrices by $A, B \in \mathbb{R}^{n \times m}$, such that $A_{ij}$ (resp. $B_{ij})$) gives Player 1's payoff when Player 1 plays action $i$ and Player 2 plays action $j$. 
The spaces of mixed strategies for each player are given by $\Delta^n$ and $\Delta^m$ respectively, where $\Delta^n$ is the probability simplex $\{ x \in \Rplus^n | \sum_i^n x_i = 1\}$. For a vector $x \in \mathbb{R}^d$, we denote its support by $supp(x) = \{ i \in [d] | x_i \neq 0 \}$. 
The \textit{best-response} of Player 2 against Player 1's (mixed) strategy $x \in \Delta^n$ is any pure strategy $\text{BR}_2(x) = \argmax_{y \in [m]} x^T A y$, where ties are broken arbitrarily. When $y$ is a non-negative integer, we adopt the shorthand $x^T A y = x A e_y$, where $e_i$ is the $i$-th elementary basis vector. Best-responses of Player 1 are similarly denoted by $\text{BR}_1(y)$. 

\paragraph{Nash equilibrium (NE) in zero-sum games}
Recall that in a zero-sum game, $A=-B$. 
Finding a Nash equilibrium $(x^*,y^*)$ in a two-player zero-sum game can be formulated as a saddle-point problem using the payoff matrix $A$ of the first player.

In particular, von Neumann's \textit{minimax theorem} shows that $(x^*, y^*)$ is achieved by the saddle point of $x^T A y$,
\begin{align}
    \max_{x\in\Delta^n} \min_{y \in \Delta^m} x^T A y = \min_{y \in \Delta^m} \max_{x\in\Delta^n} x^T A y = {x^*}^T A y^*.
\label{eq:minimax-theorem}
\end{align}
The minimax theorem shows that $(x^*, y^*)$ is a NE if and only if each mixed strategy optimizes the players' payoff assuming that the other player best responds. 
Two-player zero-sum games can be solved efficiently using a variety of methods, including fictitious play \cite{brown:fp1951}, linear programming \cite{shoham2008multiagent}, and self-play with no-regret learners \cite{hart2000simple}.

\paragraph{Sparse commitments in zero-sum games}
We now consider a scenario where Player 1 is limited to playing strategies with support at most $k$. For the rest of the paper, we refer to Player 1 (resp. Player 2) as the sparse player (resp. the non-sparse player) and vice versa.
Player 1's restricted strategy space is given by
\begin{align*}
\Delta^n_k = \{ x \in \Delta^n \text{ such that } |supp(x)| \leq k \},
\end{align*} 
Then Player 1's set of optimal bounded-support strategies is
\begin{align}
    x_k^* = \argmax_{x\in\Delta^n_{k}}\min_{y\in[m]} x^TAy. 
    \label{eq:sparse-commitment-def}
\end{align}
Note that $\Delta^n_{S}$ 
is not convex due to the cardinality constraint. Thus, the minimax theorem \eqref{eq:minimax-theorem} does not hold. Instead, we have the weaker result $\max_{x\in\Delta^n_k}\min_{y\in\Delta^m} x^TAy \leq \min_{y\in\Delta^m}\max_{x\in\Delta^n_k} x^TAy$.

\begin{remark}
In our definition, only one player is subjected to sparsity constraints; the other (i.e., inner) player is allowed to play any mixed strategy it desires, even one which is not sparse. This assumption has no bearing on our results, since at least one of Player 2's best responses will be pure.
\end{remark}
\paragraph{Strong Stackelberg equilibrium (SSE) in general-sum games}
In general-sum two-player games, where $A$ and $B$ may be arbitrary $n \times m$ matrices, a Stackelberg equilibrium $(x^*, y^*)$ can be formulated as a bilevel problem~\cite{von2004leadership,conitzer2006computing}:
\begin{align*}
    x^* = \argmax_{x\in\Delta^n} x^TAy^*(x), \text{where} \quad
    y^*(x) = \argmax_{y\in[m]} x^TBy.
\end{align*}
The equilibrium is considered \textit{strong} if, in addition, Player 2 breaks ties in favor of Player 1. In two player zero-sum games, NE and SSE coincide for Player 1, i.e., the optimal strategy $x^*$ that Player 1 commits to is also a NE strategy. For this paper, we always assume that Player 1 is playing the role of the Stackelberg leader (outer optimization problem). 

The SSE applies to general-sum games, always exists, enjoys a unique payoff, and can be computed in polynomial time, making it a popular choice of equilibrium for security applications \cite{sinha2018stackelberg,tambe2011security}.

\paragraph{Sparse commitments in general-sum games} Just like zero-sum games, we define optimal $k$-sparse commitments via the restricted strategy spaces $\Delta^n_k$. The strategy $x^*_k$ is then 
\begin{align*}
    x_k^* = \argmax_{x\in\Delta_k^n} x^TAy^*(x), \text{where} \quad
    y^*(x) = \argmax_{y\in \text{BR}_2(x)} x^T Ay.
\end{align*}

It is important to differentiate between \textbf{natural sparsity} in game solutions, which refers to the \textit{existence} of equilibria with small support, and the \textbf{enforced sparsity} in our context, which typically involves \textit{trade-offs} in terms of strictly reduced performance for the sparse player. Although many structured games (e.g., extensive form games or other games played on graphs) exhibit naturally sparse equilibria, these are typically not sparse enough to be practically applicable.

\section{Structure of Sparse Equilibria and the Limits of Na\"{\i}ve Sparsification Methods}\label{sec:theory}

We now explore pathological properties of sparse equilibria which will justify the necessity of our proposed method over more intuitive or heuristic approaches.  
For one, it is tempting to believe that the solutions $x^*_k$ and $x^*_{k'}$ for $k < k'$ are ``close'' in some sense. Unfortunately, this is  not the case: $supp(x^*_k)$ and $supp(x^*_{k'})$ can be completely disjoint.

\begin{proposition}
    There exist zero-sum games where the optimal sparse commitments $x_{k}^*$, $x_{k'}^*$ for $2 \leq k<k'\leq \sqrt{2n}$ have disjoint supports, i.e.,
    $\text{supp}(x_{k}^*)\cap\text{supp}(x_{k'}^*)=\emptyset.$
    \label{thm:disjoint_support}
\end{proposition}

Despite Proposition~\ref{thm:disjoint_support}, one may hope that even if $supp(x^*_k)$ and $supp(x^*_{k'})$ differ greatly, player 1's expected utility under them may be somewhat nicely behaved. Define $v^*_k=\max_{x\in\Delta^n_k}\min_{y\in[m]} x^TAy$, the reward to the first player under the optimal size $k$ commitment. Clearly, $v^*_k$ is non-decreasing in $k$. 
However, the utility of the first player is not necessarily ``concave'' in $k$, i.e., the marginal returns from increasing $k$ is not necessarily diminishing. 
\begin{proposition}[non-diminishing marginal returns] 
    There exist zero-sum games where $v^*_{k+1} - v^*_k$ is not non-increasing. 
    \label{thm:nonconvex_utility}
\end{proposition}
\begin{proposition}[non-submodularity]
    Let $S \subseteq [n]$ and $\Delta_S^n = \{ x \in \Delta^n | x_i = 0 \quad \forall i \not \in S \}$. Let $x^*_S = \argmax_{x \in \Delta^n_S} \min_{y \in [m]} x^T A y$ be the optimal commitment by Player 1 when restricted to playing only $S$, such that $v^*_S$ is a set function (of $S$) denoting its corresponding utility. There exist zero-sum games where $v^*_S$ is not submodular in $S$.
    \label{thm:nonsubmodular}
\end{proposition}

We also remark that the counterexamples in Propositions~\ref{thm:disjoint_support} and \ref{thm:nonconvex_utility} are not ``over-engineered'', we observe similar phenomena in our experiments.
The above propositions have consequences for game solving. It is known that the problem of finding the optimal sparse strategy is NP-hard, even when the NE or the SSE can be found in polynomial time~\cite{mccarthy2018price}. 
Propositions~\ref{thm:disjoint_support}, \ref{thm:nonconvex_utility} and \ref{thm:nonsubmodular} go further, implying a combinatorial structure incompatible with some popular game solving approaches. 

One such approach is that of \textit{incremental strategy generation}, also known as \textit{oracle-based methods}. A popular variant known as the \textit{double-oracle} method is depicted in the Appendix. These methods begin with a restricted (usually very small) set of actions, 
then an  equilibrium is computed for the subgame restricted to those actions, 
and the subgames are iteratively expanded by querying 
\textit{best-response oracles} for new actions to add.
This is guaranteed to converge in finite though potentially exponential time \cite{zhang2024exponential}. Oracle methods exploit the observation that in real-world games best responses can often be found or approximated efficiently, and usually the number of iterations of double oracle is low.

At first glance, incremental strategy generation seems appealing in our setting, as it naturally produces sparse strategies if the algorithm terminates before the restricted action set becomes larger than $k$. 
Unfortunately, it turns out that this approach, while guaranteed to produce $k$-sparse strategies, can yield low-quality commitments. For example, Proposition~\ref{thm:disjoint_support} implies that the restricted set needs to be of size $\mathcal{O}(k^2)$ if it wants to include all optimal $1,2,\dots, k$-sparse strategies. In fact, Proposition~\ref{thm:nonconvex_utility} shows that for very large games (where the true NE may not be easily computed), one may be misled into prematurely terminating if increasing $k$ does not yield an improvement in $v^*_k$. Finally, Proposition~\ref{thm:nonsubmodular} implies that a greedy-based approach that greedily expands the allowed support of Player 1's commitment (one at a time) is not guaranteed to yield a good approximate commitment, even for zero-sum games.

So far, our statements have been about the supports of the sparse equilibrium. Can we say something about the probabilities that each action is played? For instance, the algorithm of \citet{mccarthy2018price} proposes finding ``operationalizable'' commitments where Player 1 is restricted to playing $k$-uniform strategies. Unfortunately, it turns out that $k$-uniform strategies can perform badly for small $k$. 

\begin{proposition}
    Let $\bar{\Delta}^n_k = \{ x \in \Delta^n | x_i = m/k, m \in \mathbb{Z} \}$ be the set of randomized strategies played uniformly over exactly $k$ strategies (possibly repeated). Define the corresponding sparse equilibrium
    $\bar{x}^*_k = \argmax_{x \in \bar{\Delta}^n_k} \min_{y \in [m]} x^T A y$ and its value $\bar{v}^*_k$. 
    Then, (i) $\bar{v}^*_k$ is no longer non-decreasing in $k$, and
    (ii) there exist games where finding an $\epsilon$-NE operationalizable strategy $\bar{x}^*_k$ requires $k$ to grow in a rate $\Omega(1/(\epsilon+c))$, where $c$ is an adjustable game parameter.  
    \label{thm:unif-is-bad}
\end{proposition}

Proposition~\ref{thm:unif-is-bad} suggests that finding optimal sparse equilibrium is not just about reasoning about which actions belong to the optimal support; one should also jointly reason about the probabilities with which these strategies are played. In the Appendix we show an instance where our definition of a sparse commitment gives a vastly different (and higher quality) solution compared to that of \citet{mccarthy2018price}. Notably, while \citet{mccarthy2018price} argue that uniform strategies are easier to operationalize --- our example shows that this comes at a price that could be large.
\section{Finding Optimal $k$-sparse Commitments}

The non-convex set $\Delta_k^n$ (a union of polytopes) can be represented via the following set of mixed-integer constraints: 
\begin{align*}
    1&=\sum_{a\in [n]} x(a), \quad k \geq \sum_{a\in [n]} z(a), \\ 
    z(a) &\geq x(a), \quad x(a) \in [0,1],\quad z(a) \in \{ 0, 1 \} &&\forall a\in [n]. \notag
\end{align*}
This gives a mixed integer linear program (MILP) for computing $k$-sparse equilibria, which we call the \textit{basic method}:
\begin{align}\label{mip:vanilla-zerosum}
\max_{g\in\real,x\in\Delta^n_k} ~g, \quad g \leq\sum_{a\in [n]} A(a,b)x(a) &&\forall b \in [m]. \tag{B}
\end{align}
Here, $g$ represents the payoff which we are seeking to maximize, while the constraint upper bounds $g$ over all possible opponent actions; such a formulation can be obtained by \textit{dualizing} the inner minimization problem in \eqref{eq:sparse-commitment-def}.
The basic method can be implemented easily via off the shelf MILP solvers. However, computation becomes prohibitively slow when $n$ or $m$ is large. We present methods to alleviate this.

\paragraph{(i) Single oracle methods when $m$ is large but $n$ is small}

When only Player 2 (the non-sparse player) has a large strategy space, we propose a \textit{single-oracle} approach. 
For this, we assume access to Player 2's best response oracle $\text{BR}_2(x)$. We start with a subset of strategies $\mathcal{M}_0\subseteq [m]$. At each iteration $i$, we compute the optimal commitment $x_k^{\mathcal{M}_i}$ against the non-sparse player's strategy space $\mathcal{M}_i$. A best-response strategy $\text{BR}_2(x_k^{\mathcal{M}_i}))$ is then computed for the non-sparse player against the sparse player's strategy $x_k^{\mathcal{M}_i}$, using the best-response oracle. The best-response is added to the non-sparse player's strategy space, i.e. $\mathcal{M}_{i+1}\leftarrow \mathcal{M}_{i}\cup\text{BR}(x_k^{\mathcal{M}}))$, and the MIP~\eqref{mip:vanilla-zerosum} is resolved. This iterative process continues until the equilibrium gap, defined as \(\nabla = g - u(x_k^{\mathcal{M}},\text{BR}(x_k^{\mathcal{M}}))\), where $g$ is the sparse player's optimal expected utility computed by~\eqref{mip:vanilla-zerosum}, is below a pre-specified threshold \(\epsilon > 0\). When the algorithm terminates, we are guaranteed an $\epsilon$-optimal $k$-sparse commitment.

\paragraph{(ii) MILP-based method when $n$ is large but $m$ is small}
When only $n$ is large, the basic method has a prohibitively large number of integer variables. However, in many cases, the huge action space arises because of some \textit{combinatorial structure}, e.g., in security applications, the action space for patrollers is often the set all length $L$ paths in a directed graph. While $n$ is extremely large, we may exploit this combinatorial structure to compute $k$-sparse commitments. We term these generally as \textit{MILP-representable} action spaces,
which have actions given by the nonempty set $\mathcal{K} = \{ z \in \{ 0, 1 \}^l | Fz = a \}$, for constants $F \in \mathbb{R}^{n\times l}, a \in \mathbb{R}^n$. Clearly $\mathcal{K}$ is finite, but extremely large. However, we can rewrite a $k$-mixture of strategies more concisely as:
$\mathcal{S} = \{ x^{(i)} \in \mathbb{R}^m, z^{(i)} \in \mathcal{K}, t \in \Delta^k | x^{(i)} \geq t^{(i)} - (1-z^{(i)}), x^{(i)} \leq t^{(i)}, x^{(i)} \leq z^{(i)}, x^{(i)} \geq 0 \}$.
Essentially, $\mathcal{S}$ ``duplicates'' $\mathcal{K}$ up to $k$ times, where $x^{(i)}$ is some $z^{(i)}$ ``rescaled'' by the probability $t^{(i)}$ with which it is played. The constraints enforce that $x^{(i)}=t^{(i)} \cdot z^{(i)}$. For a given pure strategy $z \in\mathcal{K}$, we define the payoff function for all opponent strategies $b \in [m]$: 
\begin{align*}
    A_{z, b} = \max \{\mathcal{C}_{z, b}\}, \text{where }\mathcal{C}_{z, b} = \{ c \in \mathbb{R} | C_b z + d_b \geq 1 c\}.
\end{align*}
Here, $C_b \in \mathbb{R}^{r \times m}$, $d_b \in \mathbb{R}^r$ are game constants for every action $b \in [m]$.
Rescaling $\mathcal{C}_{z, b}$ to account for $t^{(i)}$ yields: 
\begin{align*}
    &\max_{\substack{g\in\real,(x,z,t)\in\mathcal{S}, c \in \mathbb{R}^{k \times m}}}~g \\  
    &g \leq \sum_{i \in [k]} c(i, b),  
    C_bx^{(i)} + d_b \cdot t^{(i)} \geq 1 c(i, b) 
    \quad \forall b \in [m].
\end{align*}
While somewhat complicated, allowing the cost function to be \textit{non-linear} function of $z$ admits a much wider class of problems, including games involving security interdiction \cite{vcerny2024layered,zhang2017optimal}. As a concrete example, we return to the patrolling example where Player 1's action is to choose a length $L$ path. Suppose Player 2's action is to select a subset $b$ of exactly $w$ distinct edges to hide evidence of criminal activity, where $w$ is a game constant. Player 1 obtains a payoff of $1$ if and only if its chosen path intersects at least $1$ of Player 2's chosen edges; no double counting is allowed. Then $z^{(i)}$ is a binary vector of length $\ell$ (the number of graph edges) such that $z^{(i)}(e) = 1$, $z^{(i)}(e)=t^{(i)}$ for all edges $e$ included in path $i$. For a path $z$, $A_{z,b}$ is \textit{not} a linear function of $z$, since path $z$ can intersect the set of edges chosen in $b$ \textit{multiple times}. Yet, by using the above formulation we can set $r=2$, and define the first row of $C_b$ to contain $1$ on every edge included in action $b$, $d_b(1) = 0$, second row of $C_b=0$, and $d_b(2)=1$. Therefore the constraint forces $\mathcal{A}_{z,b}$ to be upper bounded by $1$ and the number of times path $z$ intersects the edges in $b$, avoiding the issue of double counting.
\paragraph{(iii) Combined methods when $n$ and $m$ are large}

The case when both $n$ and $m$ are large can be handled by combining methods (i) and (ii). Naturally, we would require Player 2 to be able to compute best responses to any mixed strategy ``efficiently'' for any Player 1 strategy with constant-sized support.  This approach also requires solving the $k$-sparse MILP from Section (ii) ``somewhat-efficiently.'' 

\subsection{Extensions to gen-sum Stackelberg equilibria}
\label{subsec:finding:sse}

The \textit{multiple LPs} method~\cite{conitzer2006computing} is a standard approach for finding optimal commitments. Combined with our method, it yields $k$-sparse general-sum Stackelberg equilibria by solving \textit{multiple MILPs}, one for each action $b\in[m]$ of Player 2. Specifically, we solve the following MILP for each action:
\begin{align*}
\label{multiple_milp}
\max_{x\in\Delta_k^n} \tag{G}
&\sum_{a\in[n] }A(a, b)x(a) \\
    0 &\geq\smashoperator{\sum_{a\in[n]}}\left(B(a, b')-B(a, b)\right)x(a) &&\forall b' \in [m]. 
\end{align*}
The constraint guarantees $x$ is restricted 
such that $b$ is a best response to $x$. Note that for some choices of $b$ this MILP could be infeasible, though at least one is feasible. At the end, we select the maximum value achieved over all of the MILPs that are feasible. We describe practical speedups and an alternative single-MILP method in the appendix.

\subsection{Extensions to structured sparsity constraints}

For some security applications, we may want Player 1's sparsity to be imposed in a more structured manner. 
For example, suppose we are planning patrols in a graph, such that actions are paths of a given length $L$. Suppose further that we are also required to decide the locations of $k_1$ and $k_2$ command posts at which paths are allowed to begin and end. Thus, we do not require sparsity over the path distribution, but rather sparsity in the \textit{number of starting and ending vertices} over length $L$ paths chosen with non-zero probability. Consequently, utilizing $\Delta^n_k$ is overly restrictive, as there may be paths sharing a starting or ending (or both) vertex.

More generally, we will define $R$ \textit{constraint sets} $\mathcal{S} = \{S^1,\dots,S^R\}$. For each constraint set indexed by $i \in [R]$, we have \textit{action sets} $S^i = \{ S^i_{1}, \dots, S^i_{q_i} \}$, where each $S^i_{j} \subseteq \mathcal{A}$. Each constraint set has a specified sparsity constraint $k_i$. In the above patrolling example, $R=2$, one for starting points, and another for ending points. Then $q_1$ (resp. $q_2$) is equal to the number of vertices, and $S^1_{j}$ (resp. $S^2_{j}$) contains all length $L$ paths starting from (resp. ending at) vertex $j$. 

Note that the sets in $S^i$ do not necessarily form a partition of $\mathcal{A}$; for a given constraint set $i$, an action \( a \in \mathcal{A} \) can either be in one, multiple, or no action sets. 
For every action set, we incur a sparsity cost as long as at least one action it contains is played with positive probability. This yields the MILP:

\begin{align*} \label{mip:structured_space}
1 &= \sum_{a\in [n]}x(a), \quad && x(a) \in [0,1]  \quad\forall a\in [n] \\
k_{i} &\geq \sum_{S^i_j \in S^i} z(S^i_j) &&\forall i \in [R] \tag{S}
 \\  
z(S^i_j)&\geq x(a) \quad &&\forall a \in  S^i_j, \: S^i_j  \in S^i, \: i \in [R] \\         
 z(S^i_j)&\in\{0,1\} && \forall S^i_j  \in S^i, \: i \in [R]. 
 \end{align*}

\section{Empirical Evaluation}\label{sec:exp}

\begin{figure*}[ht]
    \centering
    \begin{subfigure}[t]{0.245\linewidth}
    \centering
    \includegraphics[width=\textwidth]
    {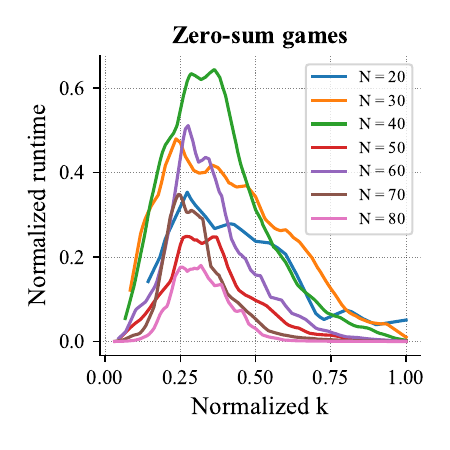}
\end{subfigure}
\hfill
\begin{subfigure}[t]{0.245\linewidth}
    \centering
    \includegraphics[width=\textwidth]{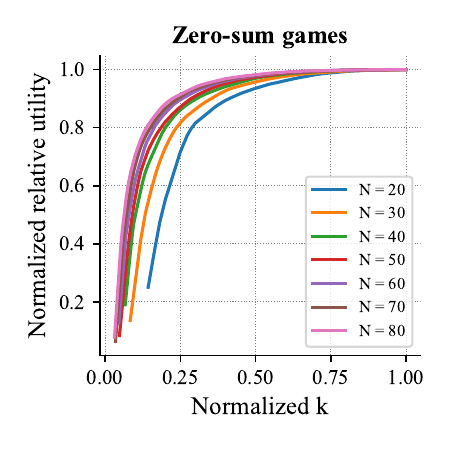}
\end{subfigure}
    \begin{subfigure}[t]{0.245\linewidth}
    \centering
    \includegraphics[width=\textwidth]{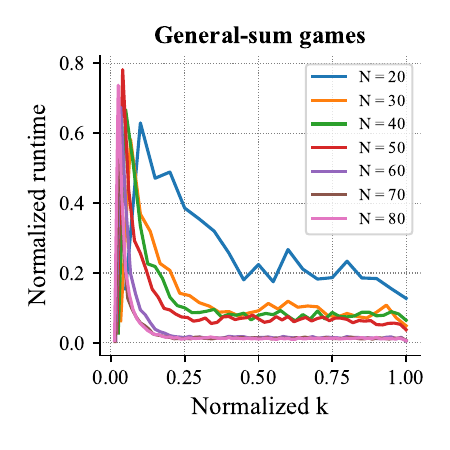}
\end{subfigure}
\hfill
\begin{subfigure}[t]{0.245\linewidth}
    \centering
    \includegraphics[width=\textwidth]{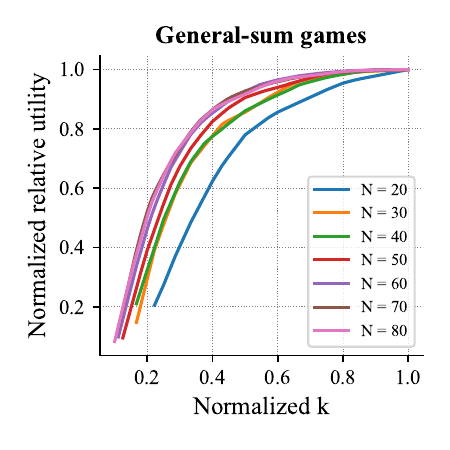}
\end{subfigure}
    \caption{Average normalized runtime and relative utility for solving random zero-sum (\textbf{left}) and general-sum (\textbf{right}) games.}
    \label{fig:rgg:gensum}
\end{figure*}

\begin{figure*}[t]
    \centering
    \begin{subfigure}[t]{0.246\linewidth}
    \centering
    \includegraphics[width=\textwidth]
    {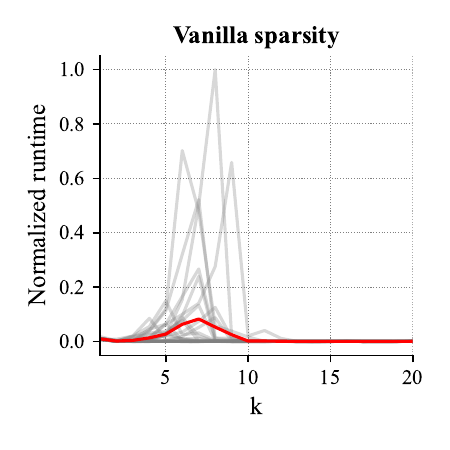}
\end{subfigure}
\hfill
\begin{subfigure}[t]{0.246\linewidth}
    \centering
    \includegraphics[width=\textwidth]{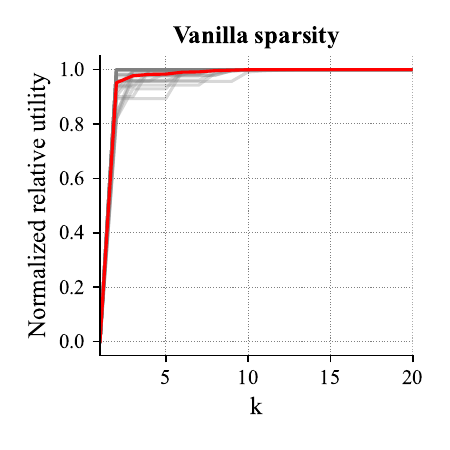}
\end{subfigure}
\begin{subfigure}[t]{0.246\linewidth}
    \centering
    \includegraphics[width=\textwidth]{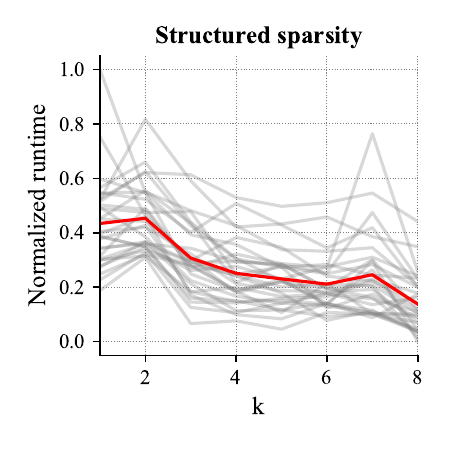}
\end{subfigure}
\hfill
\begin{subfigure}[t]{0.246\linewidth}
    \centering
    \includegraphics[width=\textwidth]{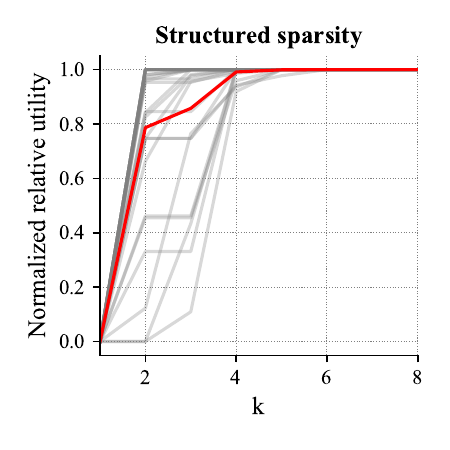}
\end{subfigure}
    \caption{Average (red) normalized runtime and relative utility over 30 instances (gray) of a PE on a university campus. We impose \textit{vanilla sparsity} on path distribution (\textbf{left}) and \textit{structured sparsity} on the starting point distribution of Player 1 (\textbf{right}).}
    \label{fig:lgg:vanilla_structured}
\end{figure*}

We empirically analyze the performance and scalability of our  proposed methods, focusing on (i) wall-clock computational time, (ii) the utility that Player 1 achieves for different support sizes $k$, and (iii) comparisons to $k$-uniform strategies.
For purposes of comparison, we normalize running times using $t_{\text{normalized}} = \frac{t-t_{\min}}{t_{\max}-t_{\min}}$, where $t_{\min}$ and $t_{\max}$ are the minimum and maximum runtimes respectively. We normalize the utility value by the Nash value for zero-sum games and the Stackelberg equilibrium (SE) for general-sum games. Specifically, $u_{\text{normalized}} = \frac{u-u_{\min}}{u_{\text{equi}}-u_{\min}}$ where $u_{\min}$ is the utility at $k=1$ and $u_{\text{equi}}$ is the true value, i.e. $k=\infty$. We use Gurobi \cite{gurobi} for all MILP solvers.
Unless otherwise stated, for experiments with randomness, 30 instances were generated and solved. We report the standard errors in plots (usually negligible). Other experimental details are deferred to the Appendix.
\subsection{Fully enumerable strategy spaces}
We begin by analyzing enforced sparsity on games where $n$ and $m$ are small, under two different experiment setups. 
\paragraph{Randomly generated normal-form games} 
We consider both zero-sum and general-sum games and solve them using the proposed MIP~\eqref{mip:vanilla-zerosum} and multiple MILP~\eqref{multiple_milp} respectively. 
We generate payoff matrices $A$, $B$ with sizes $n=m \in \{ 20,30,\dots,80\}$. For zero-sum games, each $A_{ij}$ is uniformly chosen in $[10,100]$.
In general-sum games, $A_{ij}$ and $B_{ij}$ are randomly chosen in $[-50, 50]$, under the condition that $A_{ij}$ and $B_{ij}$ have opposite signs.

Figure~\ref{fig:rgg:gensum} shows the normalized runtime and relative utility as a function of $k$, each averaged over all instances per game size. For visualization purposes, we normalized the values of $k$ with respect to the support size of NE (resp. SE) for zero-sum (resp. general-sum) games, denoted by $k_{NE}$ (resp. $k_{SE}$). Thus, the $x$-axis represents the enforced support size w.r.t. the optimal strategy. 
Interestingly, the runtime curves exhibit a notable phase transition, i.e., the difficulty is concentrated in a specific range of $k$ values. For example, for zero-sum games we find a ``hard zone'' when $k$ lies between $0.2 k_{NE}$ and $0.8 k_{NE}$.  
For general-sum games, this hard zone was observed for values of $k\leq 0.2 k_{SE}$. This pattern is reminiscent of the behavior observed when solving random SAT instances~\citep{gent1994sat}. 
Furthermore, the \textbf{utility curves} demonstrate that as the game size increases, we are still able to capture nearly $90\%$ of the unrestricted optimal commitment, even with a fraction of the support size. 
When $n=m=80$, about $90\%$ of the game value is attained at around $0.25 k_{NE}$ and $0.45 k_{SE}$. 

\paragraph{Patrolling games} We evaluate our MIP~\eqref{mip:vanilla-zerosum} on a security motivated example using the layered graph framework of \citet{vcerny2024layered}. We design a pursuit evasion (PE) game on a university campus: a professor is trying to apprehend a student who stole exam questions. 
We constructed the game using a real university campus path network. Each player starts at one of several potential initial locations on campus corresponding to vertices, and selects a path of length $T=5$ to traverse (also known as \textit{depth}); the student is apprehended if and only if the chosen paths share an edge. 
We call a set of fixed starting points and desirable final locations a \textit{scenario}. Individual games are then generated by assigning random (uniform from $[6,10]$ for desirable and $[1,5]$ for all others) rewards for escaped students based on their final location.

Figure~\ref{fig:lgg:vanilla_structured} shows results across 30 instances when we impose sparsity on the distribution over the professor's paths.  
Just like random games, we again observe a phase transition pattern, with a pronounced ``hard zone'' when $k$ ranges from 5 to 10. Furthermore, the professor is able to achieve more than $90\%$ of the optimal even for $k$ as low as 3. 

\paragraph{Air defense battery placements} Our task is to select 
which and where to place air defenses to maximize defensive coverage. Player 1's actions involve choosing (a) which, out of $4$ different types of air defense batteries; those which provide stronger coverage have lower range and vice versa, and (b) where to place it's chosen battery. These locations are based off maps from the real world. 
Player $2$ chooses a location to attack; the probability of success depends on how the defensive converges there. $k$-sparsity is sensible since batteries and ammunition dumps are expensive to procure and construct.
Due to space constraints, we report detailed results in the Appendix.
Nonetheless, our main observations are similar: while natural sparsity is quite low (e.g., $<20$), we can still do well even with a small $k$ (around $5$). We also remark that the pathological behavior in Propositions~\ref{thm:disjoint_support} and \ref{thm:nonconvex_utility} are seen in practice in these experiments.
 
\begin{figure}[t]
    \centering
   \includegraphics[width=1\linewidth]{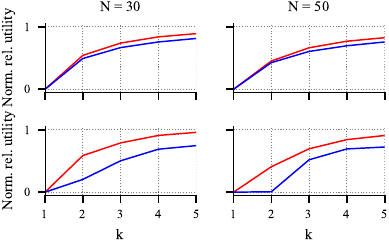}
\caption{Average normalized relative utility for solving randomly generated zero-sum and general-sum games using optimal $k$-sparse (red) and $k$-uniform (blue) commitments. \textbf{Top:} zero-sum games. \textbf{Bottom:} general-sum games.}

\label{fig:rdm_comp}
\end{figure}

\subsection{Fully enumerable structured strategy spaces} 

We evaluate MILP~\eqref{mip:structured_space} on experiments requiring structured sparsity, using the patrolling game (with $T=5$) as the basis for two tests. In the first experiment, the professor chooses paths from $8$ starting points but is limited to using $k$ starting points, potentially with multiple paths per starting point. The results are reported in Figure~\ref{fig:lgg:vanilla_structured}. Interestingly, we no longer observe a clear phase transition. Again, we obtain almost maximum utility at just $k=4$, i.e., around half of the possible starting points. Secondly, we experimented on a more complex setting with sparsity constraints $k_1$ and $k_2$ on starting points and paths respectively. The runtime heat map in Figure~\ref{fig:heatmaps} reveals a ``2 dimensional hard zone'', with the hardest instances occurring at $k_1$ between $[4,7]$ and $k_2$ in $[4,6]$.

\begin{figure}[t] 
   \includegraphics[width=0.23\textwidth]{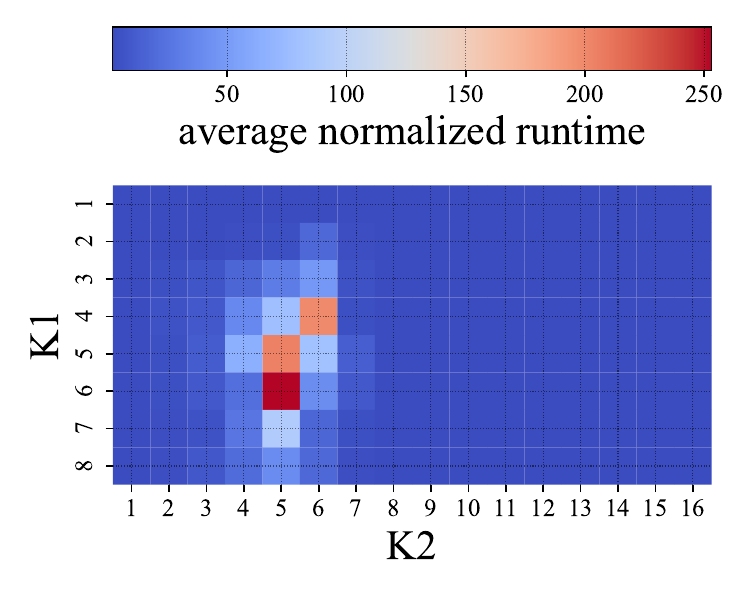}
   \includegraphics[width=0.23\textwidth]{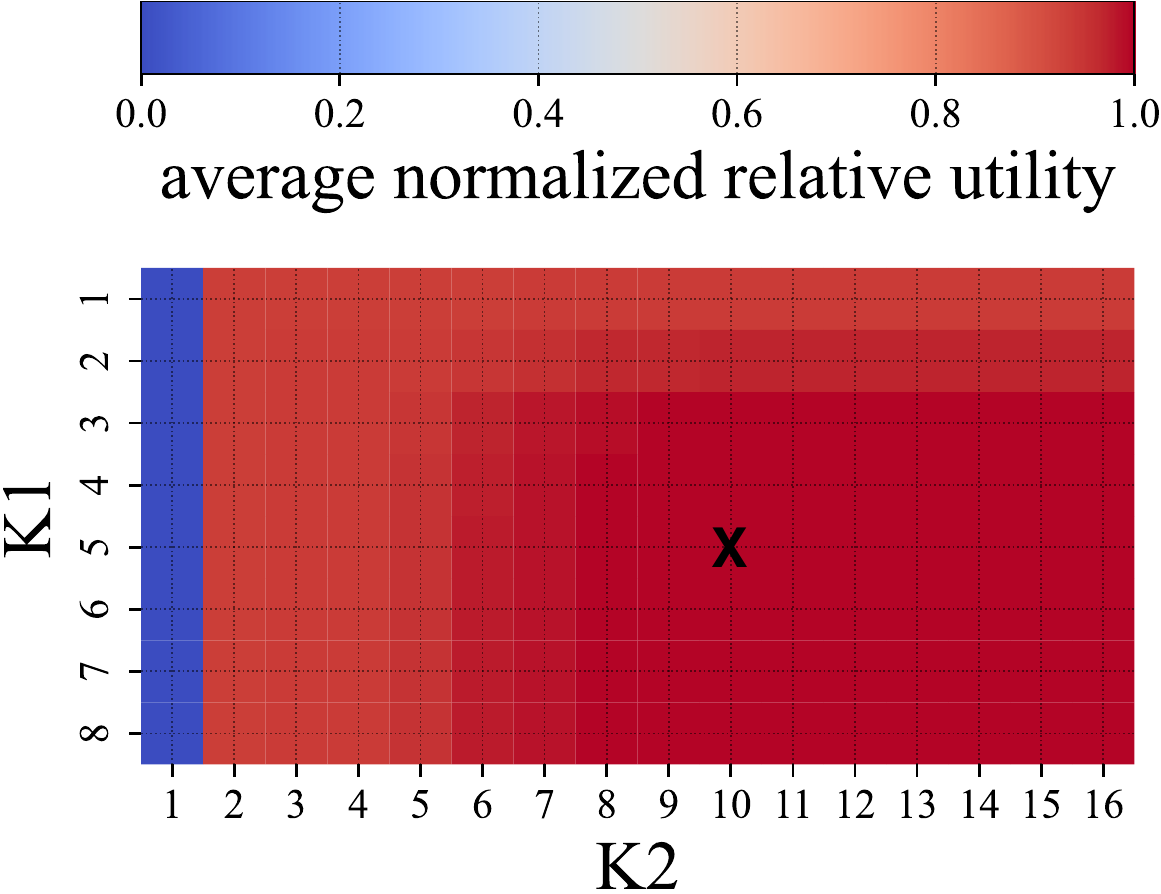}
\caption{Average normalized runtime (left) and relative utility (right) for solving a pursuit-evasion game on a university campus, imposing \textit{structured sparsity} on the distribution over starting points ($k_1$) and paths ($k_2$) of Player 1. The smallest values of $k_1$ and $k_2$ at which the Nash value is attained is marked by ``$\times$''.}
\label{fig:heatmaps}
\end{figure}

\subsection{Scaling up to large strategy spaces}
 
\begin{figure*}[ht]
    \centering
   \includegraphics[width=1\linewidth]{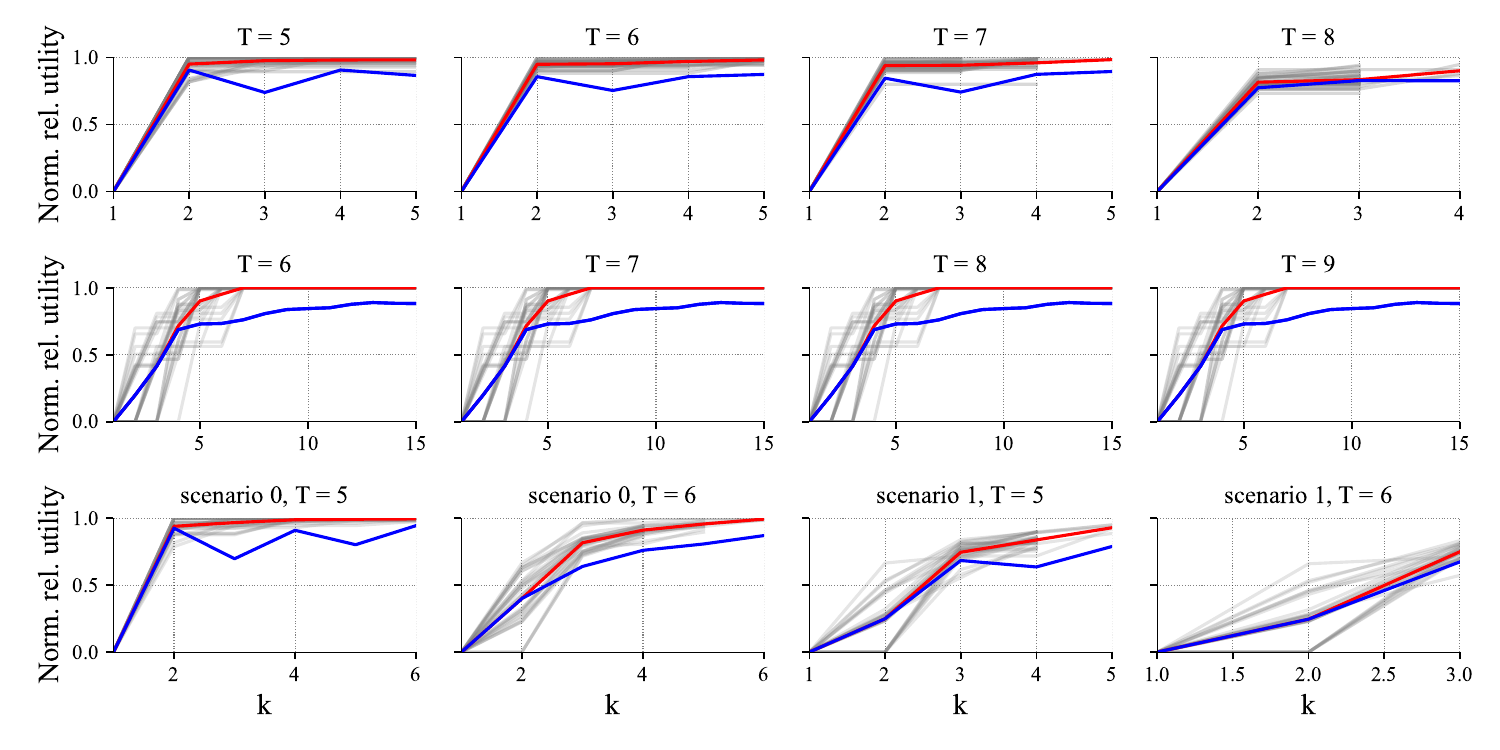}
\caption{Average normalized relative utility curves for optimal $k$-sparse (red) and $k$-uniform (blue) commitments. Results for individual instances using optimal $k$-sparse commitments are in gray. \textbf{Top row:} sparse player has a large strategy space. \textbf{Middle row:} non-sparse player has a large strategy space. \textbf{Bottom row:} both players have large strategy spaces.}
\label{fig:large_spaces}
\end{figure*}

Next, we examine settings where the basic MIP approach fails due to the large strategy space of one or both players. Using patrolling game variants, we test three setups: (i) large $n$, small $m$; (ii) small $n$, large $m$; and (iii) large $n$ and $m$. For each setup, we run 30 instances and report average relative utility results for instances solved within a 12-hour limit.
\paragraph{(i) Sparse player has a large action space} 
Here, only Player 1 has a large action space, which we increase by varying $T$ from 5 to 8 (note that $n$ grows exponentially with the path length $T$). We observe in  Figure~\ref{fig:large_spaces} (top) that for $T\in \{5,6,7\}$, Player 1 achieves more than $90\%$ of the optimal game value for $k$ as low as 3. Even at $T=8$, we obtain a respectable $80\%$ of the true value when $k=3$.
\paragraph{(ii) Non-sparse player has a large action space} In this setting, Player 2 has a large action space whereas Player 1 places only a single patrol on a single intersection in the graph, which then stays there. We vary the depth of the graph from 6 to 9. The relative utility plots reported in Figure~\ref{fig:large_spaces} (middle) show that Player 1 captures around $90\%$ of the true Nash value at $k=3$ for all depths. Notice however that the plots are almost identical. The reason behind this is that increasing depth does seem to provide Player 2 with no strategic advantage, Player 1 hence positions the patrol on the same intersections as with smaller $T$. 

\subsubsection{(iii) Both players have large action spaces}
Finally, we examine the case where both Player 1 and 2 have large action spaces. We study two scenarios differing in the locations of the players' starting points and desirable exits. We present in Figure~\ref{fig:large_spaces} the plots for $T=\{5,6\}$ across the two scenarios. In scenario 0, Player 1 gains high utility value at small $k$, notably for depth 5. Solving for the optimal $k$-sparse strategy at depth 6 in scenario 1 proves to be challenging. While some instances achieve approximately $80\%$ of the Nash value, a support size of 3 remains insufficient for attaining higher utility. We point out that here, there are $\sim$ 16k paths per player, accounting for the relatively poor performance.

\subsection{Comparison with $k$-uniform strategies} 
We compare the average relative utility achieved by Player 1 when using the optimal $k$-sparse commitment against the $k$-uniform mixed strategy proposed by \citet{mccarthy2018price}. Again, we consider random normal-form games and patrolling games with large action spaces. For the former, we generate zero-sum and general-sum games as described earlier, considering games where $n = m \in {30,50}$. We limit the $x$-axis to $k = 5$ to emphasize small support sizes, which are of primary interest. As shown in Figure~\ref{fig:rdm_comp}, the results indicate that $k$-sparse commitments generally outperform $k$-uniform strategies, especially when $k$ is small. This difference is more pronounced in general-sum games.

Figure~\ref{fig:large_spaces} depicts the results in patrolling games, where the blue curves represent the average relative utility normalized by the Nash value when the sparse player employs $k$-uniform strategies. Consistent with the findings from random games, we observe that optimal $k$-sparse strategies provide higher utility for the sparse player, particularly when the support size is small. For instance, for $T={5,6,7}$, the sparse player secures over 95\% of the Nash value with $k=3$ when both players have large strategy spaces, whereas a uniform strategy yields less than 80\%.

\section{Conclusion}

This paper studies properties and algorithms for finding $k$-sparse commitments for two player games. Our proposed method performs competitively with non-sparse commitments even for very small values of $k$. Furthermore, we have proposed extensions to handle Stackelberg commitments in general-sum games, structured sparsity, as well as cases with a large number of player actions. Our notion of $k$-sparse commitments often leads to significantly better performance when compared to $k$-uniform strategies. Future work include extensions to extensive-form or multiplayer games, exploring approximate sparsity using sparsity-inducing regularization techniques, or online learning of sparse strategies.


\section*{Acknowledgments}
This research was supported by the Office of Naval Research awards N00014-22-1-2530 and N00014-23-1-2374, and the National Science Foundation awards IIS-2147361 and IIS-2238960. Part of the work was done while CKL was a postdoctoral scientist at Columbia University.

\bibliography{ref}
\newpage
\onecolumn
\section{Appendix}
\subsection{Remarks on Sparse Approximations, Regularization, and Online Learning}

While running our algorithms to full completion yields an exact sparse solution, doing so can be prohibitively time-consuming. In practice, the algorithms may be terminated earlier by, for example, using a less stringent termination criterion in the MILP solver or selecting a larger equilibrium gap, $\nabla$. The resulting strategy is an approximate solution in that the $k$-sparse commitment is strictly enforced, but the strategy may not be entirely optimal. This approach aligns better with the practical applications we envision. Alternatively, we can consider a broader definition of approximate sparse solutions, where the sparse player's strategy is nearly sparse -- specifically, no more than $k$ actions have a probability of at least $\epsilon$. Such strategies could potentially be identified, e.g.,  using sparsity-inducing regularization. However, applying nonconvex regularization alone does not ensure a $k$-sparse equilibrium as per our definition; achieving this may require an additional truncation or projection step, which introduces further approximation error.

Another interesting use case for finding or implementing sparse strategies is in online learning. Defining sparsity in an online setting, however, presents additional challenges. Perhaps the most natural extension is to require that the online learning algorithm always selects an iterate from the $k$-sparse simplex at every iteration. However, it is possible, in principle, play a pure strategy at every iteration, for example via Follow The Perturbed Leader (FTPL).

\subsection{Double Oracle Algorithm}

In this section we present a pseudo code of a standard incremental strategy generation method described in Section~\ref{sec:theory}.
\def\blueplan{x}
\def\redplan{y}
\begin{algorithm}[h]
\caption{Double Oracle Algorithm}\label{alg:do}
\begin{algorithmic}[1]
\STATE $\mathcal{N}, \mathcal{M} \gets \textsc{InitialSubgame}([n], [m])$ \label{algline:init_subgame}
\REPEAT
\STATE $\widetilde{x}^*, \widetilde{y}^* \gets \textsc{NashEquilibrium}(\mathcal{N}, \mathcal{M})$ \label{algline:solve_subgame}
\STATE $\blueplan^\text{BR}, \redplan^\text{BR} \gets \textsc{BR}_1(\widetilde{y}^*), \textsc{BR}_2(\widetilde{x}^*)$ \label{algline:best-response-oracle}
\STATE $\mathcal{N}, \mathcal{M} \gets \mathcal{N}\cup \{\blueplan^\text{BR}\}, \mathcal{M}\cup \{\redplan^\text{BR}\}$ \label{algline:expand-br}
\UNTIL{$\textsc{EquilibriumGap}(\widetilde{x}^*, \widetilde{y}^*, \blueplan^\text{BR}, \redplan^\text{BR}) < \epsilon$}
\end{algorithmic}
\end{algorithm}

\subsection{Single General-Sum Stackelberg Equilibrium MILP} 
\label{appendix:single_milp}
Here we formulate a single MILP that condenses the entire multiple LP approach described in Section~\ref{subsec:finding:sse} in one mathematical program. We note that in practice, the multiple LP method is often faster. This MILP has $n+m$ binary variables, $n\times m + n$ real variables, and $2n+3m$ constraints.

\begin{align*}
\max &\sum_{a\in[n]}\sum_{b\in[m]}A(a, b)r(a,b) \\
    0 &=s_{b'} + \sum_{a\in[n]}B(a, b')x(a) - \sum_{a\in[n]}\sum_{b\in[m]}B(a, b)r(a,b) &&\forall b' \in [m] \\
    y(b) &= \sum_{a\in [n]}r(a,b) && \forall b\in [m] \\
    1 &= \sum_{a\in [n]}x(a) \\
    1 &= \sum_{b\in [m]}y(b) \\
    z(a)&\geq x(a) && \forall a\in [n] \\ 
    S &\geq \sum_{a\in [n]} z(a) \\ 
    0 &\leq r(a,b) \leq x(a) && \forall a\in [n]\\
    0 &\leq s_{b} \leq M(1-y(b)) && \forall b\in [m]\\
    0 &\leq x(a) \leq 1 && \forall a\in [n] \\
    0 &\leq r(a,b) \leq 1 && \forall a\in [n],\; \forall b\in [m] \\
    z(a)&\in\{0,1\} && \forall a\in [n] \\
    y(b)&\in\{0,1\} && \forall b\in [m]
\end{align*}

\subsection{Proof of Proposition~\ref{thm:disjoint_support}}
We demonstrate a class of games where the optimal sparse commitments for $k=2,\dots$ have completely disjoint support. Let $N > 2$ be a constant.
    The actions for each player are identical: $\mathcal{A}_1 = \mathcal{A}_2 = \mathcal{A}$, and thus $n=m=|\mathcal{A}_1|=|\mathcal{A}_2|$. 
    The actions are partitioned into $N$ subsets
    \begin{align*}
        \mathcal{A} = \bigcup_{i=2}^{N} \mathcal{A}^{(i)},
    \end{align*}
    where the $i$-th subset contains $i$ distinct elements as follows:
    \begin{align*}
        \mathcal{A}^{(i)} = \{ a_{i,0}, a_{i,1}\dots, a_{i,i-1} \};
    \end{align*}
    where $i\in [2,\dots,N]$ and the $a_{i, \alpha}$'s are distinct actions such that  $|\mathcal{A}^{(i)}| = i $ and $n=m=|\mathcal{A}| = \Theta(N^2)$.

    As we will soon see, we will engineer the payoffs such that if both players are restricted to be playing actions in some fixed $\mathcal{A}^{(i)}$, then we will recover an `$i$-th order matching pennies' game (plus some constant payoff to Player 1 which we will elaborate on later). That is, if Player 2 correctly guesses which of the $i$ actions Player 1 chooses, Player 1 receive a payoff of $-r \cdot (i-1)$, but if they do not, Player 1 receive a payoff of $r$.  Therefore, under this restricted action subgame, $\mathcal{A}^{(i)}$ the optimal strategy is for both players to play uniformly at random. Player 2 then guesses right $1/i$ times and wrong $1-(1/i)$ times, and we recover a fair game: the expected payoff is $ r \cdot \left( -(1-i) \frac{1}{i}  + \frac{i-1}{i}\right) = 0$.

    We now move back to the original game where actions can be played across subsets.
    Our utility function is decomposed into 2 components,
    $$U(a_{i\alpha}, a_{j\beta}) = U_\text{base}(a_{i\alpha}) + U_\text{fight}(a_{i\alpha}, a_{j\beta}),$$
    where $U_\text{base}$ is a positive payoff that Player 1 always gets no matter what the opponent (Player 2) plays, and in fact, only depends on the action set that Player 1's action belongs to. It is given by
    \begin{align*}
    U_\text{base}(a_{i\alpha}) = 
    \frac{2^{i-1}}{2^N} = 2^{i-1-N} < 1,
    \end{align*}
    and $U_\text{fight}$ is an additive component that gives $0$ utility if both players play in the different action sets $x$ and $y$, but an extremely high-stakes (though still fair) high order matching pennies if they play in the same game. Specifically, 
    \begin{align*}
    U_\text{fight}(a_{i\alpha}, a_{j\beta}) = 
    \begin{cases}
        0 \qquad & i \neq j \\
        -r \cdot (\alpha-1) \qquad & i = j \wedge \alpha = \beta \\ 
        r \qquad & i =j \wedge \alpha \neq \beta.
    \end{cases}
    \end{align*}
    We claim that for large enough $r = \Omega(N^2)$, the best $k \leq N$-sparse commitment for Player 1 (as the max player) is to play $a_{k,0}, a_{k,1}, \dots a_{k,k-1}$ uniformly at random, which gives them an expected payoff of $2^{k-1-N}$, where the component from $U_{\text{fight}}$ is equal to zero under the opponent's best response. 

    First, observe that as far as best responses go, Player 2 is indifferent between all of their actions. Obviously, only $U_\text{fight}$ can be influenced by Player 2. Recall that we only need to consider pure strategy best responses. If they play some action in $\mathcal{A}^{(j)}$ where $j\neq k$, then they get a payoff of zero for the $U_\text{fight}$ component. On the other hand, if they play some $a_{k\beta}$ in the $k$-th action set, then the fact that Player 1 is playing uniformly in that action set means that Player 2 will get $0$ regardless of which $a_{k\beta}$ they choose (recall the structure of $k$-th order matching pennies). Thus, once we account for $U_\text{base}$, it is clear that $2^{k-1-N}$ is a lower bound that Player 1 can guarantee among all $k$-sparse commitments for all $k\leq N$. We will prove that Player 1 can do no better; in fact, this is the \textit{unique} $k$-sparse commitment.
    \begin{theorem}
        The unique optimal $k$-sparse commitment for the above game for $2 \leq k \leq N$ is to play actions in $\mathcal{A}^{(k)}$ uniformly at random.
        \label{thm:best-k-sparse-commitment}
    \end{theorem}
    Theorem~\ref{thm:best-k-sparse-commitment} obviously implies Proposition~\ref{thm:disjoint_support}. We will prove Theorem~\ref{thm:best-k-sparse-commitment} using 2 Lemmata, which when combined will rule out Player 1 placing any mass on any action set outside of $\mathcal{A}^{(k)}$ in any optimal commitment.
    \begin{lemma}
        Suppose $r$ is sufficiently large. Then for mixed strategies of Player 1 where the total mass $M$ played in actions $\bigcup_{w=k+1}^N \mathcal{A}^{(w)}$ is strictly greater than $0$, there exists a best response giving Player 1 a payoff \textbf{strictly} less than $2^{k-1-N}$.
        \label{thm:counterexample-toolarge}
    \end{lemma}
    \begin{proof}
    We split the final payoff into 4 components. $\hat{u}_{\text{base}}$, $\check{u}_{\text{base}}$, $\bar{u}_{\text{base}}$, and $u_{\text{fight}}$. The first three are for the components of $U_\text{base}$ for actions in $\bigcup_{w=k+1}^N \mathcal{A}^{(w)}$, $\bigcup_{w=1}^{k-1} \mathcal{A}^{(w)}$ and $\mathcal{A}^{(k)}$, respectively. Therefore, for a \textbf{fixed} Player 2 action $a_{j\beta}$, and distribution over actions $x$ for Player 1, we have the expected utility 

    \begin{align}
        u_\text{total} &= 
            \underbrace{\sum_{a \in \bigcup_{w=k+1}^N \mathcal{A}^{(w)}} U_{\text{base}}(a) \cdot x(a)}_{\hat{u}_\text{base}} + 
            \underbrace{\sum_{a \in \bigcup_{w=1}^{k-1} \mathcal{A}^{(w)}} U_{\text{base}}(a) \cdot x(a)}_{\check{u}_\text{base}} +
            \underbrace{\sum_{a \in \mathcal{A}^{(k)}} U_{\text{base}}(a) \cdot x(a)}_{{\bar{u}_\text{base}}}+
            \underbrace{\sum_{a \in \mathcal{A}} U_{\text{fight}}(a, a_{j\beta}) \cdot x(a)}_{{u}_\text{fight}}.
    \end{align}
    
    We first show that $\hat{u}_\text{base} + u_\text{fight} < 0$ under the best response (or for that matter, some action that we will specify in a moment). Since Player 1 places a mass $M > 0$ in actions $\bigcup_{w=k+1}^N \mathcal{A}^{(w)}$, by the pigeonhole principle, at least one $\mathcal{A}^{(w^*)}$ has mass greater or equal to $M/n$. Fix that $w^*$. Also, because Player 1 is playing a $k$-sparse strategy, at most $k$ of these actions in $\mathcal{A}^{(w^*)}$ are played with non-zero probability. Thus, again by the pigeonhole principle, at least one action is played with at least $\frac{M}{n} \cdot \frac{1}{k}$ probability. A best-responding opponent can play that action all the time. So, the probability of losing is at least $\frac{M}{n} \cdot \frac{1}{k}$. This gives a payoff of 
    \begin{align} 
    \frac{M}{N} \cdot \left( \underbrace{\frac{1}{k} \cdot -r \cdot (w^*-1)}_{\text{payoff from losing}  \leq -r}
    + \underbrace{(1-\frac{1}{k})\cdot r}_{\text{payoff from winning} \leq r \cdot (1-\frac{1}{N})} \right) \leq -Mr/N^2,
    \label{eq:fat-enough}
    \end{align}
    where the first inequality is because $k \leq w^*-1$. Recall that $U_{\text{base}} < 1$. Hence, the amount obtained from $U_\text{base}$ from playing in $\bigcup_{w=k+1}^N \mathcal{A}^{(w)}$ is no more than $M$ in total. Noting that so if we select $r$ large enough such that $Mr/N^2 > M$, the total utility from this strategy is strictly $<0$ and we are done. This can be done by selecting $r>N^2$.    

    Now, \textit{under the same choice of best response}, we will show that $\check{u}_\text{base} + \bar{u}_\text{base} < 2^{k-1-N}$. We have
    \begin{align*}
    \underbrace{\sum_{a \in \bigcup_{w=1}^{k-1} \mathcal{A}^{(w)}} U_{\text{base}}(a) \cdot x(a)}_{\check{u}_\text{base}} +
    \underbrace{\sum_{a \in \mathcal{A}^{(k)}} U_{\text{base}}(a) \cdot x(a)}_{{\bar{u}_\text{base}}} \leq 
    (1-M) \cdot \left(2^{k-1-N} \right) < 2^{k-1-N}.
    \end{align*}
    Adding all 4 terms up, we get that $u_\text{total} < 2^{k-1-N} + \text{(a strictly negative term)} < 2^{k-1-N}$ as desired.
    \end{proof}
        
    \begin{lemma}   
    If Player 1 plays a mixed strategy $x$ where the total mass $M'$ played in actions $\bigcup_{w=1}^{k-1} \mathcal{A}^{(w)}$ is strictly greater than $0$, then under any Player 2's best response Player 1 will obtain \textbf{strictly} less than $2^{k-1-N}$.
    \label{thm:counterexample-toolarge2}
    \end{lemma}
    \begin{proof}
    Let $a_{i^*\alpha^*}$ be the action in $\bigcup_{w=1}^{k-1} \mathcal{A}^{(w)}$ played with the highest probability. By our assumption, $x(a_{x^*i^*}) > 0$. 

    As in our proof of Lemma~\ref{thm:counterexample-toolarge}, we decompose the utility for any Player 2 action $a_{j\beta}$ as
    \begin{align*}
        u_\text{total} &= 
            \underbrace{\sum_{a \in \bigcup_{w=k+1}^N \mathcal{A}^{(w)}} U_{\text{base}}(a) \cdot x(a)}_{\hat{u}_\text{base}} + 
            \underbrace{\sum_{a \in \bigcup_{w=1}^{k-1} \mathcal{A}^{(w)}} U_{\text{base}}(a) \cdot x(a)}_{\check{u}_\text{base}} +
            \underbrace{\sum_{a \in \mathcal{A}^{(k)}} U_{\text{base}}(a) \cdot x(a)}_{{\bar{u}_\text{base}}}+
            \underbrace{\sum_{a \in \mathcal{A}} U_{\text{fight}}(a, a_{j\beta}) \cdot x(a)}_{{u}_\text{fight}}.
    \end{align*}
    
    We observe that Player 2 can at least play $a_{i^*\alpha^*}$ to force a utility of at most $0$ for the $u_\text{fight}$. 
    Also, from Lemma~\ref{thm:counterexample-toolarge} we may assume that actions in $\bigcup_{w=k+1}^N \mathcal{A}^{(w)}$ are played with probability $0$ each. So, $\hat{u}_\text{base} = 0$. Finally, since $M'=\sum_{a \in \bigcup_{w=1}^{k-1} \mathcal{A}^{(w)}} > 0$, we have 
    \begin{align*}
    \check{u}_{\text{base}} + \bar{u}_\text{base} &= 
    \underbrace{\sum_{a \in \bigcup_{w=1}^{k-1} \mathcal{A}^{(w)}} U_{\text{base}}(a) \cdot x(a)}_{\check{u}_\text{base}} +
            \underbrace{\sum_{a \in \mathcal{A}^{(k)}} U_{\text{base}}(a) \cdot x(a)}_{{\bar{u}_\text{base}}} \\
            &\leq  M' \cdot \left( 2^{k-2-N}\right) + (1-M') \cdot \left( 2^{k-1-N}\right) \\
            & <  2^{k-1-N}.
    \end{align*}
    \end{proof}

    From Lemmata~\ref{thm:counterexample-toolarge} and \ref{thm:counterexample-toolarge2}, we conclude that for Player 1 to obtain a utility $\geq 2^{k-1-N}$ Player 1 can only play actions in $\mathcal{A}^{(k)}$. However, under that restriction of only playing in  $\mathcal{A}^{(k)}$ it is obvious that the unique optimum is to play uniformly at random over all those actions (since it is a matching pennies game). This yields Theorem~\ref{thm:best-k-sparse-commitment}. 
    
\subsection{Proof of Proposition~\ref{thm:nonconvex_utility}}
We can directly use the example from the proof of Proposition~\ref{thm:disjoint_support}. Let $N \geq 5$. Recall that $v^*_k=2^{k-1-N}$. This is a convex, rather than concave function. More precisely, 
\begin{align*}
    v^*_{3}-v^*_{2} &= 2^{2-N} - 2^{1-N},\;\text{and}\\
    v^*_{4}-v^*_{3} &= 2^{3-N} - 2^{2-N} = 2\cdot (v^*_{3}-v^*_{2}).
\end{align*}
Since $v^*_{3}-v^*_{2}$ is strictly positive, this is clearly a strictly increasing function (and in fact exponentially increasing).
\subsection{Proof of Proposition~\ref{thm:nonsubmodular}}
We can again use the example from the proof of Proposition~\ref{thm:disjoint_support}. 
Let $S =\{ a_{2,0}, a_{2,1} \}$ and $S' = \{ a_{2,0}, a_{2,1}, a_{3,0}, a_{3,1} \}$. Clearly $S' \supset S$. Now recall that a set function $f(S)$ is submodular if and only if $f(S \cup \{z\}) - f(S) \geq f(S' \cup \{ z \}) - f(S')$ for all $S' \supseteq S$ and $z \not \in S$, i.e., there are decreasing ``marginal returns'' of adding $z$ to larger sets. Let $z=a_{3,2} $. 

\begin{align*}
    f(S) &= \frac{2}{2^N} &\tag{Uniform at random between $a_{2,0}, a_{2,1}$} \\
    f(S \cup \{z\}) &= \frac{2}{2^N} &\tag{Uniform at random between $a_{2,0}, a_{2,1}$}\\
    f(S') &= \frac{2}{2^N} &\tag{Uniform at random between $a_{2,0}, a_{2,1}$} \\
    f(S' \cup \{ z \}) &= \frac{4}{2^N} &\tag{Uniform at random between $a_{3,0}, a_{3,1}, z $} \\
\end{align*}
The proofs of the second to fourth lines follow from a similar reasoning as in the proof of Theorem~\ref{thm:best-k-sparse-commitment}. The idea is that unless all actions $a_{3,0}, a_{3,1}, a_{3,2}$ are all available, Player 1 should commit uniformly at random between $a_{2,0}, a_{2,1}$. This is because if some total nonzero mass is placed $M$ between $a_{3,0}$ and $a_{3,1}$, then at least one of them, $a^*$, is played with at least $M/2$ probability (pigeonhole principle). If Player 2 best responds by playing $a^*$, then \ref{eq:fat-enough} tells us that the payoff from fighting is $\leq -Mr/N^2$ while the base payoff (accrued from playing in $\mathcal{A}^{(3)}$) is no greater than $M$. For $r>N^2$, this is strictly less than $0$. Meanwhile, the base payoff from playing $a_{2,0}$ and $a_{2,1}$ is $(1-M)\cdot \frac{2}{2^N} < \frac{2}{2^N}$.

With these calculations, we can directly verify that the inequality required for submodularity is not satisfied for these choices of $S, S'$ and $z$.

\subsection{Proof of Proposition~\ref{thm:unif-is-bad}}

Consider a biased matching pennies game with a parameter $a > 0$, as depicted in Figure~\ref{fig:biased-matpen}. The utilities in this game range from $-1$ to $1$ for all values of $a$. The game has a unique Nash equilibrium, where Player 1 plays action $H$ with a probability of $\frac{2}{a+3}$, and Player 2 plays action $H$ with a probability of $\frac{a+1}{a+3}$. The value of this equilibrium is given by $v^* = \frac{1-a}{a(a+3)}$.

Now, consider a setting where Player 1 is restricted to $k$-uniform strategies as in the algorithm of \citet{mccarthy2018price}, while Player 2 responds optimally by choosing either $H$ or $T$. We examine a deviation $\epsilon > 0$ from the equilibrium value $v^*$ when Player 1 plays a $k$-uniform strategy $\bar{x}_k$:
\begin{equation}
    |\bar{v}_k - v^*|\leq\epsilon,
\end{equation}
where $\bar{v}_k$ is the value of $\bar{x}_k$ against some best-response of Player 2. If Player 2's best response is $H$, then $\bar{v}_k = \frac{p}{ak} - \frac{1}{a}(1 - \frac{p}{k})$, and manipulating the inequality gives us a lower bound
\begin{equation}
    k \geq \frac{(a+1)(a+3)}{a^2\epsilon+3a\epsilon+2a+2}.
\end{equation}
Similarly, to achieve the same deviation when Player 2 responds with $T$ (where $\bar{v}_k = -\frac{p}{k} + \frac{1}{a}(1 - \frac{p}{k})$), Player 1 needs
\begin{equation}
    k \geq  \frac{2(a+3)}{a^2\epsilon+3a\epsilon+4}.
\end{equation}
The required value of $k$ increases at least at a rate determined by the minimum of these bounds.

\begin{figure}[h]
    \centering
\begin{tikzpicture}
    \matrix[matrix of math nodes, 
            nodes in empty cells,
            nodes={minimum width=2cm, minimum height=1cm, anchor=center},
            column sep=-\pgflinewidth, row sep=-\pgflinewidth,
            ] (m) {
        & H & T \\
        H & |[draw]|\frac{1}{a},-\frac{1}{a} & |[draw]|-1,1 \\
        T & |[draw]|-\frac{1}{a},\frac{1}{a} & |[draw]|\frac{1}{a},-\frac{1}{a} \\
    };
\node[above=0.25cm] at ($(m-1-2)!0.5!(m-1-3)$){\textbf{Player 2}};
\node[rotate=90] at ($(m-2-1)!0.5!(m-3-1)+(-1.25,0)$){\textbf{Player 1}};
\end{tikzpicture}
    \caption{A biased matching pennies game with parameter $a$.}
    \label{fig:biased-matpen}
\end{figure}

This game also illustrates a case where our definition of a sparse commitment yields a significantly different (and higher quality) solution compared to the approach of \citet{mccarthy2018price}. Specifically, $k$-sparse strategies can exactly represent the equilibrium of this game already with $k=2$.

\subsection{Experimental Setup}
All the experiments on randomly generated games and structured sparsity experiments on patrolling games were conducted on a Linux 64-bit machine with 12th Gen Intel Core i9-12900, 2.4Ghz, with access to 64GB of RAM. The experiments on larger strategy spaces ran on a cluster consisting of machines with Intel Xeon Gold 6226, 2.9Ghz on a Linux 64-bit platform. Each run was restricted to 8 threads and 32GB of RAM. We used Gurobi 10.0.3 \cite{gurobi} for (MI)LPs.

The real-world physical graph of a US university campus was obtained using OSMnx \cite{boeing2017osmnx}. Algorithms using the best response oracle of the non-sparse player were set to use the equilibrium gap tolerance $\epsilon=10^{-3}$. More details about the domain construction can be found in the next section.

\subsection{Domain Descriptions and Additional Experimental Results}
 
\subsubsection{Randomly generated general-sum games}
\label{appendix:mult_milp_coeff}

We evaluate the single general-sum Stackelberg equilibrium MILP proposed in Section~\ref{appendix:single_milp}. We consider the same payoff matrices used to assess the multiple MILP~\eqref{multiple_milp} for sizes $n=m\in \{20, 30, \dots, 70\}$. As for the multiple MILP, we report the normalized runtime and relative utility
as a function of $k$, each averaged over all instances per game
size, and normalize the values of $k$ with respect to $k_{SE}$ (Figure~\ref{fig:sing_milp_none}). Once again, we observe a phase transition pattern in the runtime curves. However, the single MILP exhibits a higher runtime compared to the multiple MILP, especially for  games of sizes 30 and 40. Moreover, with the single MILP, we can still capture nearly $90\%$ of the game value with only a fraction of the support size.

\begin{figure*}[htb!]
    \centering
   \includegraphics[width=0.249\linewidth]{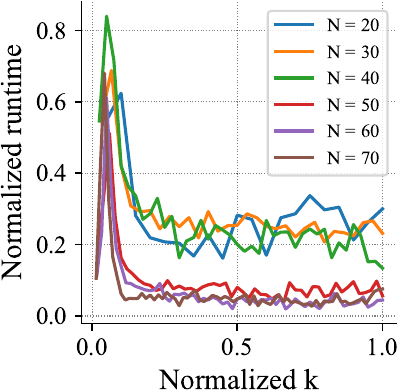}
   \includegraphics[width=0.249\linewidth]{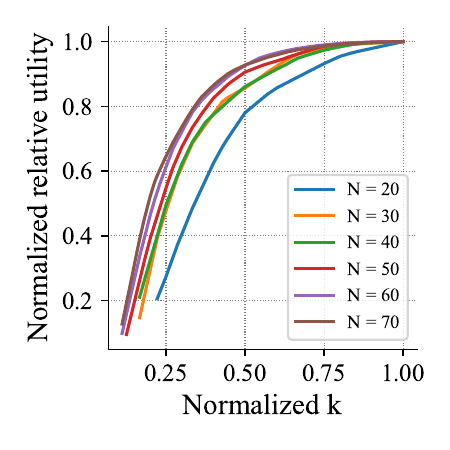}

\caption{Average normalized runtime and relative utility for solving general-sum games with the single MILP.}
\label{fig:sing_milp_none}
\end{figure*}

Additionally, we explore an alternative method for generating random payoff matrices for general-sum games. Again, the entries $A_{ij}$ are randomly chosen in $[-50,50]$. However, the entries $B_{ij}$ are now chosen such that the payoff matrix $B$ is negatively correlated with $A$ with some random noise. More precisely, we let $B_{ij}=c\times A_{ij}+N_{ij}$ where $c$ is a fixed coefficient and $N$ is a matrix of random uniform noise with entries in $[-85, 85]$. We consider different values for $c \in \{-0.8, -0.6, -0.4, -0.2\}$, and solve the games for sizes $n=m\in \{20,30,\dots,70\}$ using both the multiple and single MILPs. We present the normalized runtime and relative utility results across the different of the resulting Pearson product-moment correlation coefficients in Figures~\ref{fig:mult_milp_coeffs:runtime},~\ref{fig:mult_milp_coeffs:utility},~\ref{fig:sing_milp_coeffs:runtime}, and~\ref{fig:sing_milp_coeffs:utility}. The utility plots indicate that both MILPs exhibit very similar performance, while the runtime plots suggest that the single MILP generally has a higher runtime compared to the multiple MILP.

\subsubsection{Air defence battery placements}
We extracted maps of different regions (randomly chosen) of the world from the video game \textit{Hearts of Iron IV}. See Figure~\ref{fig:AD172} for examples. Note that the only purpose of using these maps is because land masses have been broken to different regions, we did not use utilize any other information from the video game. The optimal sparse commitments were conducted on a M1 Macbook Pro with 16 GB of RAM. The focus of this experiment is not on running times, but quality and visualization of solutions. 
We use the Basic Method (B) to compute optimal commitments.
In general, each of these experiments takes longer to run than randomly generated matrix games. That said, note that the air defence game is a fair bit larger; typically $n\approx800, m\approx 200$.

\begin{figure}[ht]
    \includegraphics[trim={1.8cm 0 1.5cm 0}, width=0.33\textwidth]{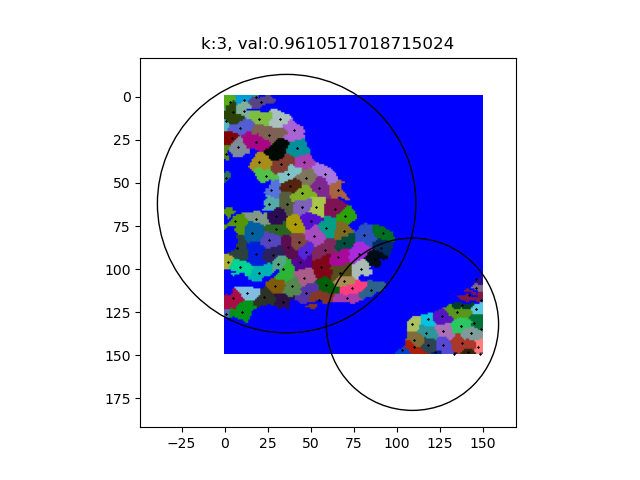}
    \includegraphics[trim={1.8cm 0 1.5cm 0}, width=0.33\textwidth]{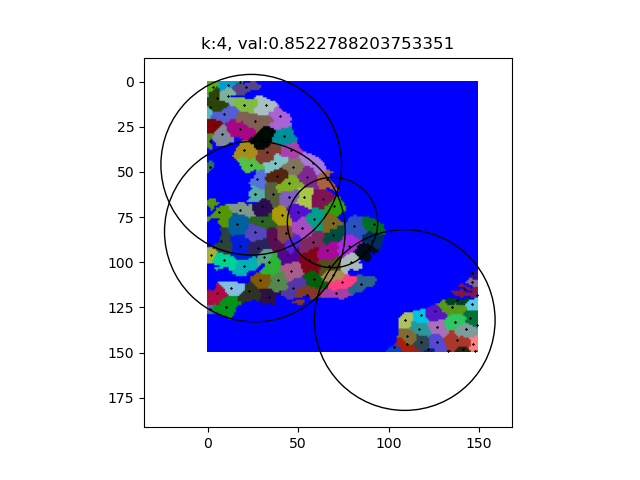}
    \includegraphics[trim={1.8cm 0 1.5cm 0}, width=0.33\textwidth]{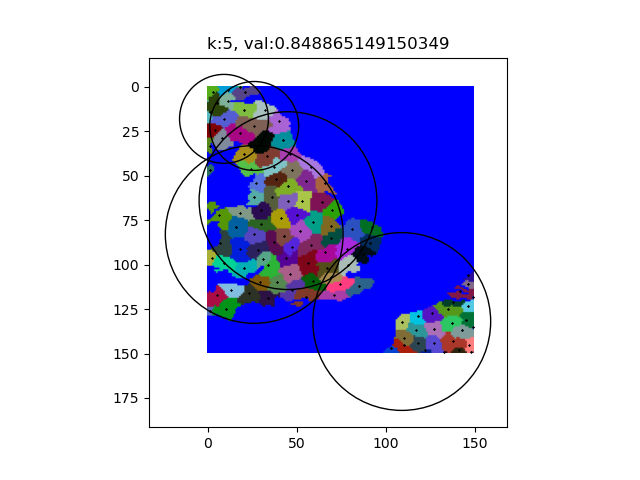}
    \\
    \includegraphics[trim={1.8cm 0 1.5cm 0}, width=0.33\textwidth]{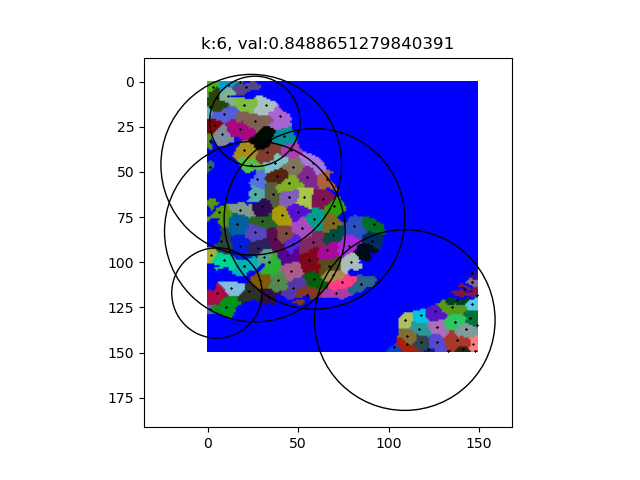}
    \includegraphics[trim={1.8cm 0 1.5cm 0}, width=0.33\textwidth]{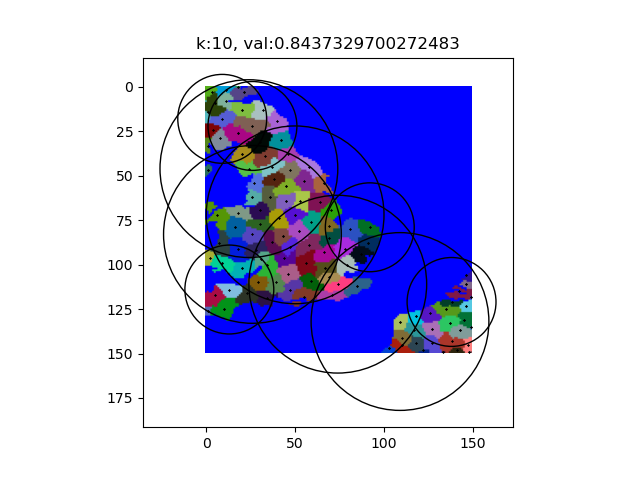}
    \includegraphics[trim={1.8cm 0 1.5cm 0}, width=0.33\textwidth]{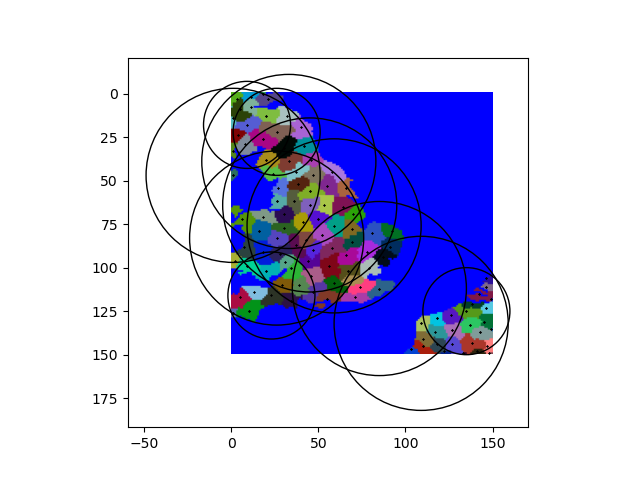}
\caption{Results for air defence coverage near a region in the United Kingdom over different values of $k$. Colored, non-blue regions are land regions, and blue shades are sea/oceans. Circles are coverage circles (out of 4 different sizes available) for locations that are chosen with strictly non-zero probability. `val' refers to the expected damage taken under that value of $k$. Bottom right picture are results for $k=\infty$.}
\label{fig:AD172}
\end{figure}

For each instance, we allow for $4$ different air defense batteries. Our job is to choose a single location (given by the map's regions) and select a single air defense battery to place there. Each type of battery contains a different coverage effectiveness and radius. Coverage effectiveness is given by the probability that the attack succeeds on a region covered by the air defense battery, while radius influences the size of the circles as seen in Figure~\ref{fig:AD172}. Player 2 chooses a single region to attack. If it is not covered, then the target is hit with probability 1. If it is covered, then with some probability the attack goes through; this probability is dependent on the type of air defense chosen. Each region has a value depending on whether it is a regular region or a city. Cities are worth more than regular regions. The amount that Player 2 gets is the value of Player 2's target, multiplied by the probability the attack goes through; we call this the \textit{expected damage taken} by Player 1. The game is zero-sum, so Player 1's payoff is likewise defined. Note crucially that Player 1 seeks to \textit{minimize} the expected damage taken; hence this is a min-max problem rather than a max-min one. 

In the above experiment, we choose radii of $(25, 50, 75,100)$ \footnote{these units are in game pixels, not ground distances. Even though the earth is round we did not perform any transformations. This is to keep matters simple.} and coverage effectiveness of $(0.05, 0.20, 0.85, 0.9)$ respectively. Cities were worth 1.1, while regular regions 1.0. Note that for the latter, lower numbers mean more effective defenses. Figure~\ref{fig:AD172} show some results for different values of $k$. Here \textit{val} refers to the expected payoff to \textit{Player 2}.  Qualitatively, we can see that this is reminiscent of some kind of ``covering problem'', when $k=2$, we are forced to use an effective battery to cover the mainland of the UK, but as $k$ increases, we are able to utilize more effective defenses, but place them with varying probabilities throughout the mainland. There is small coverage at the bottom right region of the English channel which is always ``serviced'' by the same battery at the same location. Nonetheless, in general we see that these optimal placements vary greatly as $k$ changes, similar in spirit to Proposition~\ref{thm:disjoint_support}. 

Finally, we report performance in Figure~\ref{fig:AD_results} for 6 different maps over 3 different experimental settings for a range of $k$. Note that these are different maps (and therefore incomparable across settings). We stress once again that Player 1 (the sparse player) \textit{is minimizing}. We can see that unlike before, not as much as the true optimum is achievable for low $k$. Nonetheless, having $k=4$ seems to yield reasonable results. We can clearly see the not non-increasing behavior as $k$ increases. We can also see in a concrete setting the non-convexity (which we max expect from a minimizing sparse player) of the expected damage taken, just as Proposition~\ref{thm:nonconvex_utility} says. It is also worth pointing out that the shapes of each plot in a graph look surprisingly similar, e.g., where the knots are. We suspect this is because these maps have regions roughly spread evenly in a plane. (since this was from a video game) Therefore, barring the existence of large water bodies and the slight variations caused by cities,  we are ultimately solving very similar ``set cover'' problems each time.  

\begin{figure}[ht]
    \centering
    \includegraphics[width=0.33\linewidth]{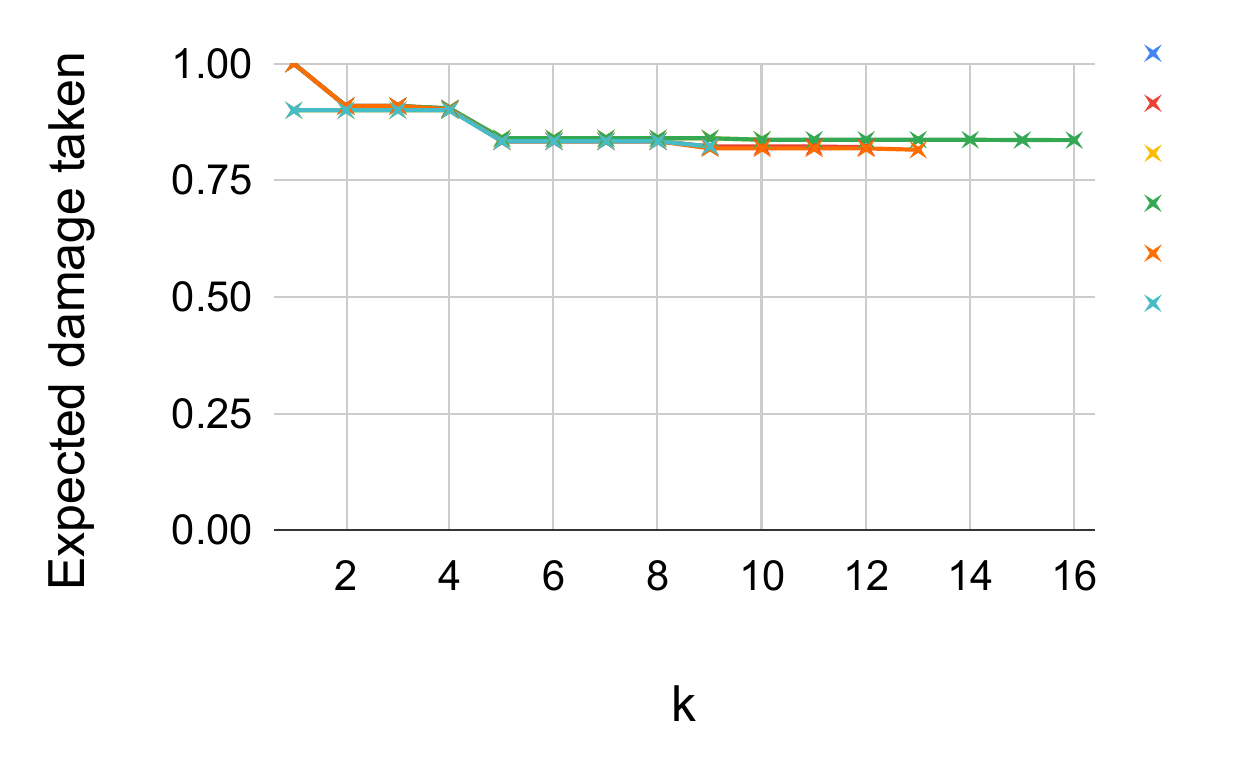}
    \includegraphics[width=0.33\linewidth]{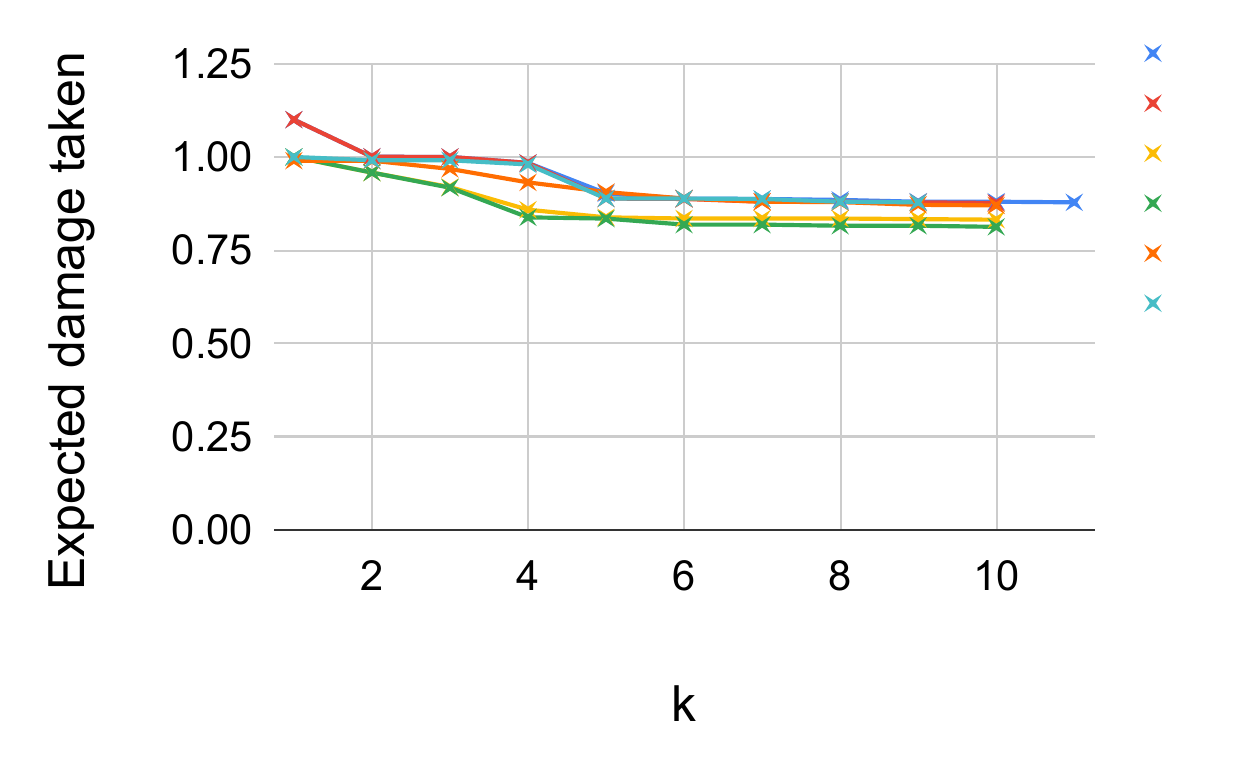}
    \includegraphics[width=0.33\linewidth]{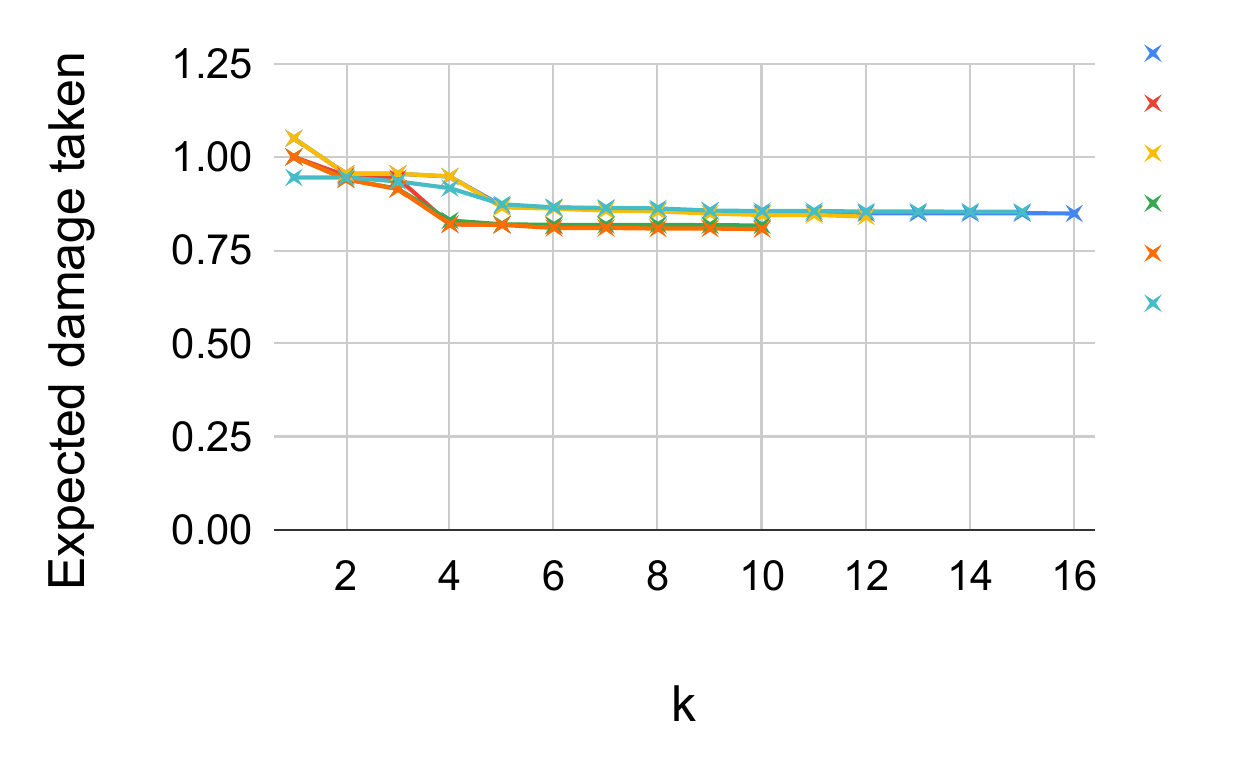} \\
    \includegraphics[width=0.33\linewidth]{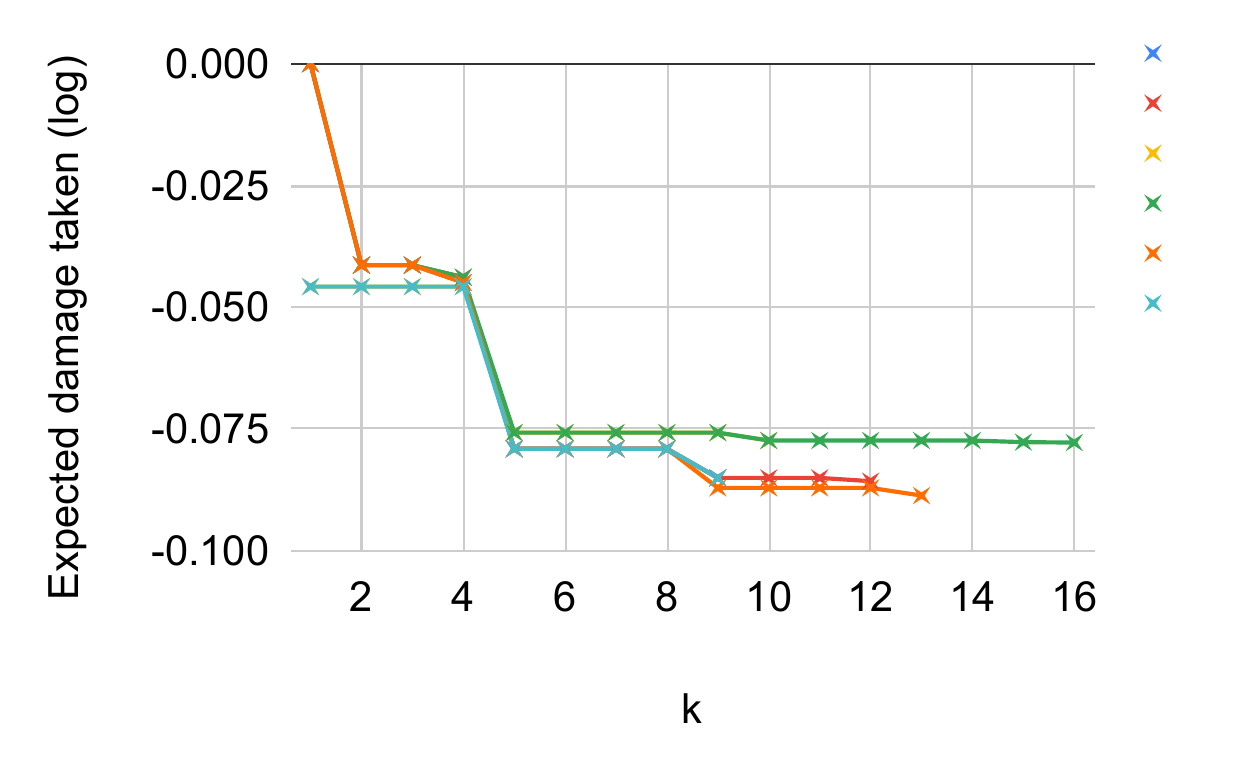}
    \includegraphics[width=0.33\linewidth]{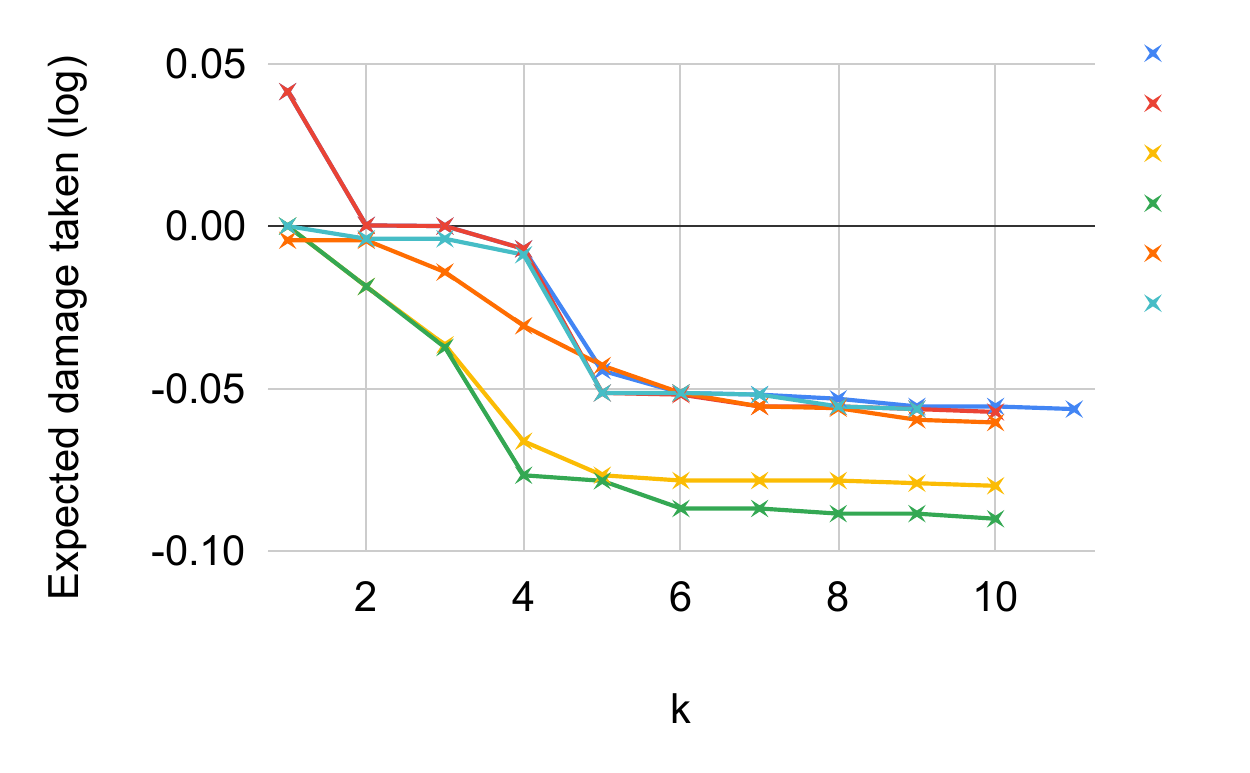}
    \includegraphics[width=0.33\linewidth]{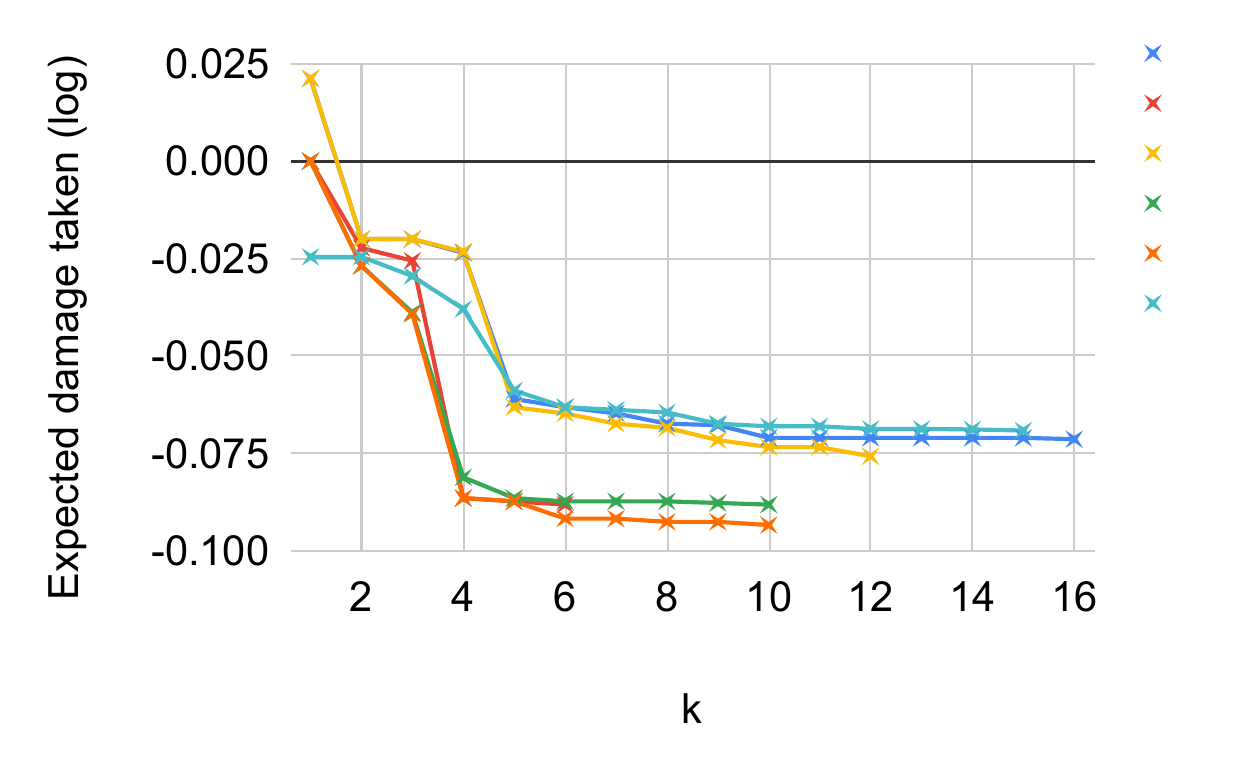}
    \caption{Expected damage taken over different $k$ for 3 settings (\textbf{top}), and in log scale (\textbf{bottom}). The radii for all air defences are $[25,50,75,100]$ for all 3 settings. Regular regions have a value of 1, and cities have a value of 1, 1.1, 1.05 from left to right. The probability of an attacker going through under coverage is (in terms of air defense types), from left to right, $[(0.01, 0.20, 0.85, 0.9), (0.05, 0.20, 0.85, 0.9), (0.05, 0.20, 0.85, 0.9)]$.}
    \label{fig:AD_results}
\end{figure}

\begin{figure}[ht]
    \centering
    \includegraphics[width=.5\linewidth]{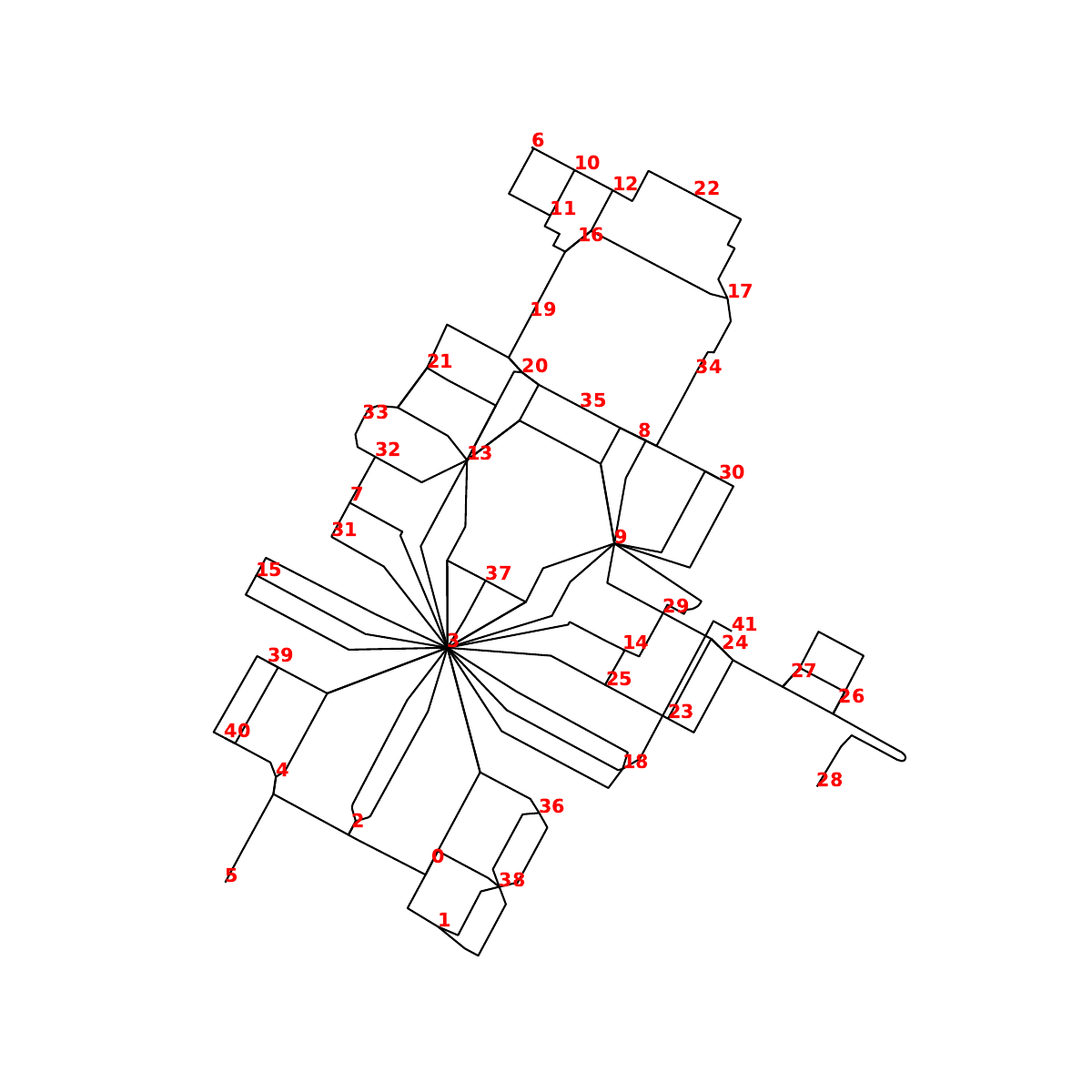}
    \caption{Physical graph of a US university campus our patrolling games are played on. Red numbers are the location indices.}
    \label{fig:columbia}
\end{figure}

\subsubsection{Patrolling games} 

The patrolling games were conceptualized as layered graph security games~\cite{vcerny2024layered} played on a physical graph, as illustrated in Figure~\ref{fig:columbia}. We assumed that each edge could be traversed in a single time step. The game was constructed by time-expanding the physical graph into a pursuit-evasion game, following the methodology from layered graph security games~\cite{vcerny2024layered}, using code available online\footnote{\url{https://github.com/CoffeeAndConvexity/LayeredGraphGamesGenerators}}. In the basic version of the game, the professor selects a path within the time-expanded graph, starting from a node within their designated starting points set. It is important to note that the number of possible paths grows exponentially with the size of the graph. Depending on the time horizon, this version was utilized for the vanilla and structured sparsity experiments (for shorter horizons) or for experiments involving a large strategy space for the professor (for longer horizons). In the scenario with a large strategy space for the student, the student selects a path originating from one of their starting points, while the professor sets up only a single checkpoint. When both players select paths, we refer to the setting as having both strategy spaces large. The specific settings for individual scenarios are detailed in Table~\ref{tab:scenarios}. In a given scenario, certain vertices in the physical graph are identified as desirable exit nodes for the student, offering higher utility than others. When generating a game for a particular scenario, the utility of a vertex is assigned uniformly at random from the interval $[6,10]$ if the vertex belongs to the set of desirable nodes, and uniformly at random from the interval $[1,5]$ otherwise.

\begin{table}[h!]
\centering
\begin{tabular}{|p{3.3cm}|p{1.5cm}|p{3.8cm}|p{3.8cm}|p{3cm}|}
\hline
\textbf{Experiment} & \textbf{Scenario} & \textbf{Starting Points Professor} & \textbf{Starting Points Student} & \textbf{Desirable Exits} \\ \hline
Vanilla sparsity          & --        &   7, 9, 14, 19, 29, 35, 36, 39                               &   --                              &  13, 30, 32, 37, 40                     \\ \hline
Structured sparsity        & --        &  7, 9, 14, 19, 29, 35, 36, 39                                &    --                             & 13, 30, 32, 37, 40                       \\ \hline
Large professor space        & Scenario 0        &  2, 4, 10, 24, 33                                &    --                            &   9, 13, 37, 35                     \\ \hline
        & Scenario 1        &     7, 9, 14, 19, 29, 35, 36, 39                             &      --                          &  13, 30, 32, 37, 40                      \\ \hline
Large student space        & --        &   --                               &   4, 9, 12, 26, 35, 40                            &  11, 13, 21, 25, 29, 35                      \\ \hline
Large both spaces        & Scenario 0        &   12, 16, 19, 21, 27, 32                               &    12, 16, 19, 21, 27, 32                            &    9, 27, 29, 33, 37                    \\ \hline
        & Scenario 1        &   4, 8, 13, 16, 22, 24                               &   4, 8, 13, 16, 22, 24                             &   0, 3, 8, 14, 35                     \\ \hline
        & Scenario 2        &   3, 5, 8, 21, 27, 39                               &    3, 5, 8, 21, 27, 39                           & 15, 20, 21, 35, 40                       \\ \hline
\end{tabular}
\caption{Settings of individual scenarios used in our patrolling game experiments.}
\label{tab:scenarios}
\end{table}

We now consider again the vanilla and sparsity experiments, but increase the path length to $T=6$, meaning that the professor selects a path of length 6.  The student still chooses a single location to escape to. Again, we impose vanilla sparsity constraints on the paths distribution and structured sparsity constraints on the starting points distribution of the professor. Figure~\ref{fig:lgg:vanilla_structured_d6} shows the results across 21 instances. The observed results are generally as expected. A phase transition pattern in the runtime curves is clearly present in the vanilla sparsity scenario but absent in the structured sparsity scenario. In both cases, however, the professor is able to achieve more than $90\%$ of the game value for $k$ as low as 3. 

\begin{figure*}[htb!]
    \centering
    \begin{subfigure}[t]{0.246\linewidth}
    \centering
    \includegraphics[width=\textwidth]
    {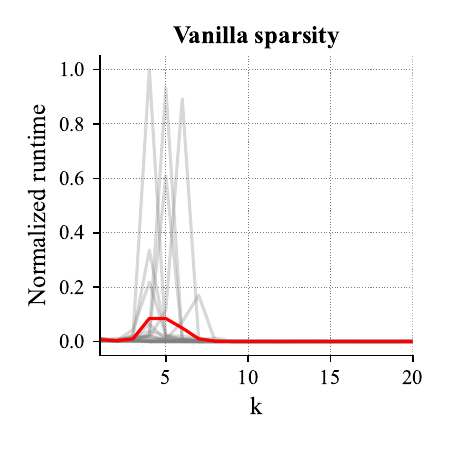}
\end{subfigure}
\hfill
\begin{subfigure}[t]{0.246\linewidth}
    \centering
    \includegraphics[width=\textwidth]{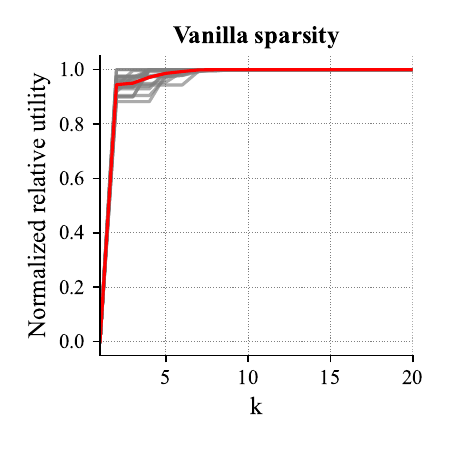}
\end{subfigure}
    \begin{subfigure}[t]{0.246\linewidth}
    \centering
    \includegraphics[width=\textwidth]{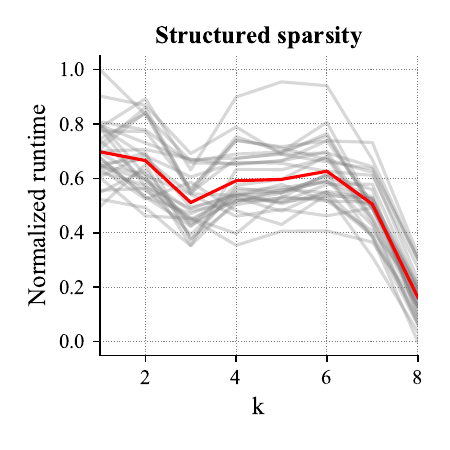}
\end{subfigure}
\hfill
\begin{subfigure}[t]{0.246\linewidth}
    \centering
    \includegraphics[width=\textwidth]{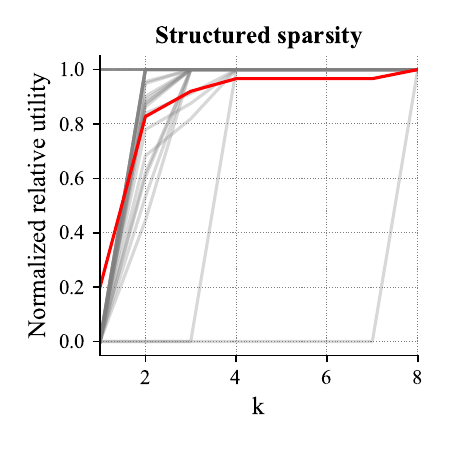}
\end{subfigure}
    \caption{Average (red) normalized runtime and relative utility over 30 instances (gray) of a PE on a university campus. We impose \textit{vanilla sparsity} on path distribution (\textbf{left}) and \textit{structured sparsity} on the starting point distribution of Player 1 (\textbf{right}).}
    \label{fig:lgg:vanilla_structured_d6}
\end{figure*}

\subsubsection{Games with large strategy spaces}

We conduct a deeper exploration of patrolling games with large strategy spaces by examining additional path lengths for each scenario and introducing new scenarios. We also provide runtime results for all experiments.

In games where the sparse player has a large action space, we evaluate a new set of starting points and desirable final locations (scenario 0) and vary $T$ from 5 to 8. The results are generally consistent with those from the previous scenario (scenario 1). However, at $T = 8$, there is a notable increase in utility as $k$ rises from 3 to 5: approximately $20\%$ of the true game value is achieved at $k=3$, $70\%$ at $k=4$, and $90\%$ at $k=5$ (Figure~\ref{fig:exp2:scen0}, top). 

For games where both players have large strategy spaces, we experiment on a new scenario (scenario 2), and consider a higher path length $T=7$ for all scenarios. The normalized relative utility results, reported in Figure~\ref{fig:exp2:scen0} (bottom), show that Player 1 is able to achieve around $90\%$ of the game value for $k=5$ and around $80\%$ for $k=3$ at $T=6,7$. The poor performance observed in scenario 1 at $T=7$ is attributed to the huge number of paths (more than 16k) combined with an extremely small support size ($k=2$). 

Moreover, we present in Figures~\ref{fig:exp2:runtime},~\ref{fig:exp3:runtime} and~\ref{fig:exp4:runtime} the runtime results for the experiments with large action spaces, across all scenarios and depths. For games where only Player 1 has a large strategy space and games where both players have large strategy spaces, we obtained results for only lower values of $k$ in most experiments within the 12-hour time limit. The average runtime tends to increase after a certain value of $k$, such as $k=6$ for scenario 0 at depth 5, when Player 1 has a large action space. However, in games where Player 2 has a large strategy space instead, the average runtime grows more slowly, allowing us to obtain results for instances up to $k=15$. Overall, the results show that runtime generally increases with $k$, though with some fluctuations.

\subsubsection{Comparison with $k$-uniform strategies} 
The blue curves in Figure~\ref{fig:exp2:scen0} represent the average relative utility normalized by the Nash value for the case where the sparse player uses $k$-sparse strategies. As expected, optimal $k$-sparse strategies generally outperform $k$-uniform strategies for $k>2$. When $k=1,2$, $k$-sparse and $k$-uniform strategies yield nearly identical utility.

\begin{figure*}[t]
    \centering
   \includegraphics[width=1\linewidth]{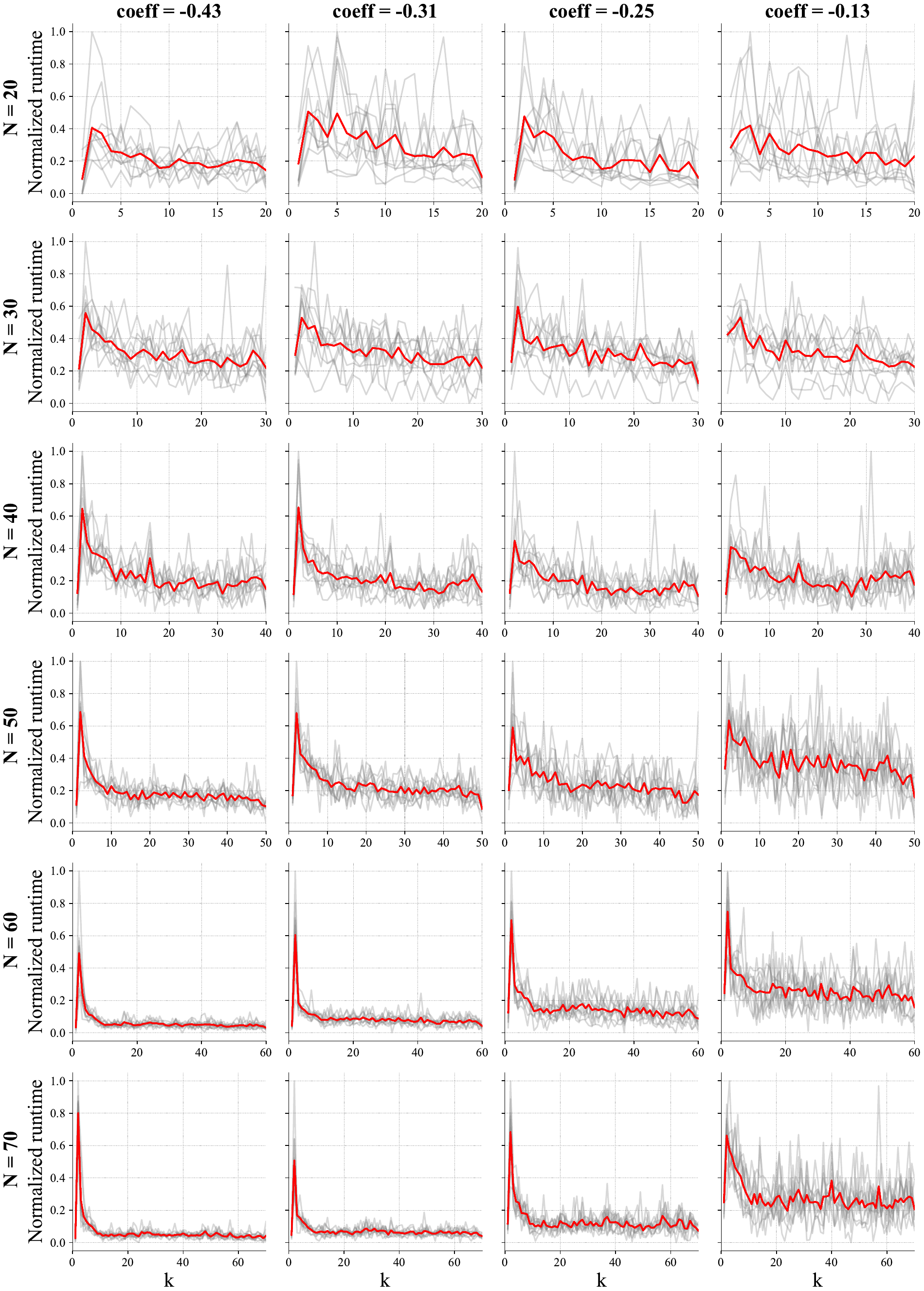}
  
\caption{Average normalized runtime across different game sizes and correlation coefficients using the multiple MILP.}
\label{fig:mult_milp_coeffs:runtime}
\end{figure*}

\begin{figure*}[t]
    \centering
   \includegraphics[width=1\linewidth]{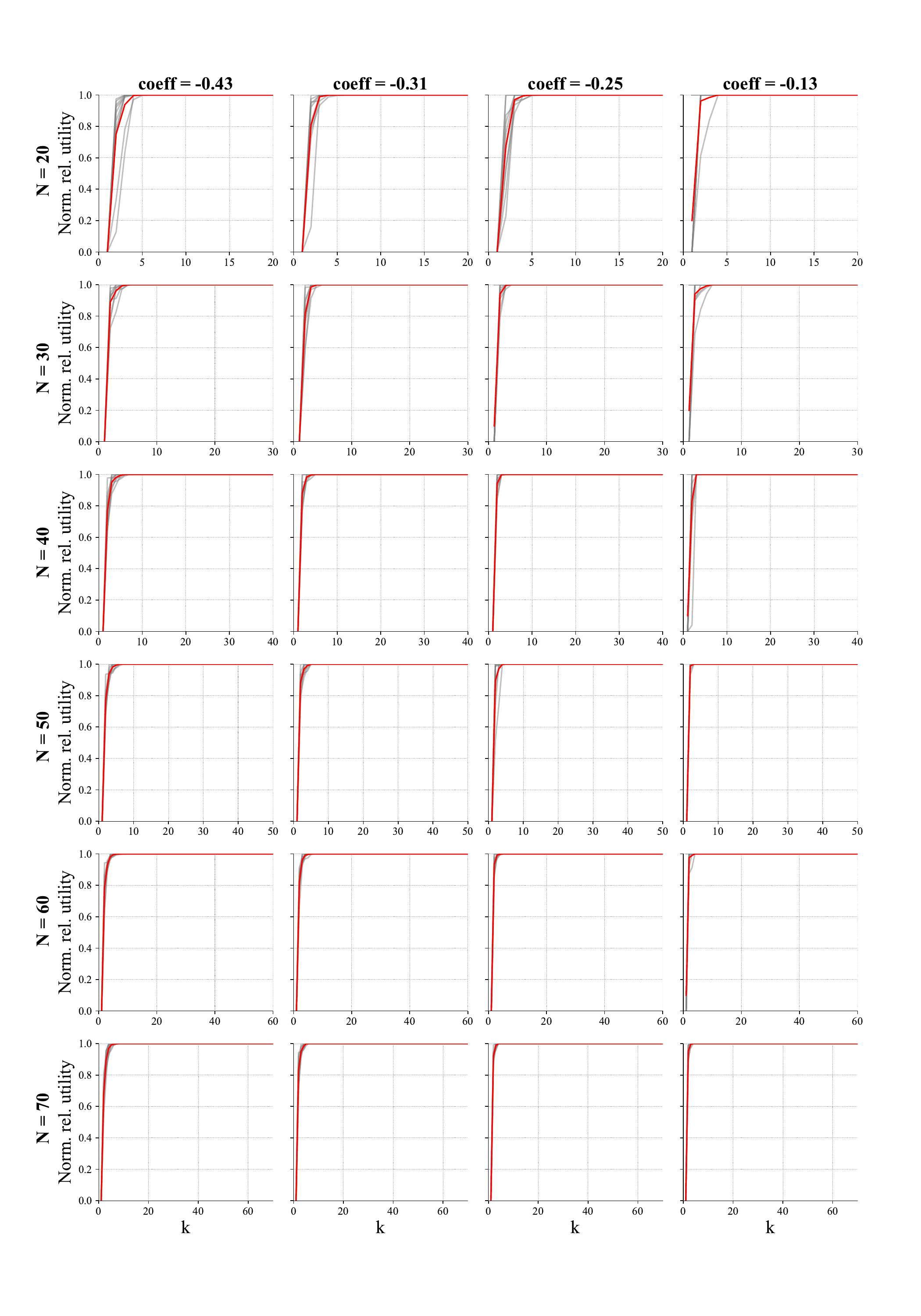}
  
\caption{Average normalized relative utility across different game sizes and correlation coefficients using the multiple MILP.}
\label{fig:mult_milp_coeffs:utility}
\end{figure*}

\begin{figure*}[t]
    \centering
   \includegraphics[width=1\linewidth]{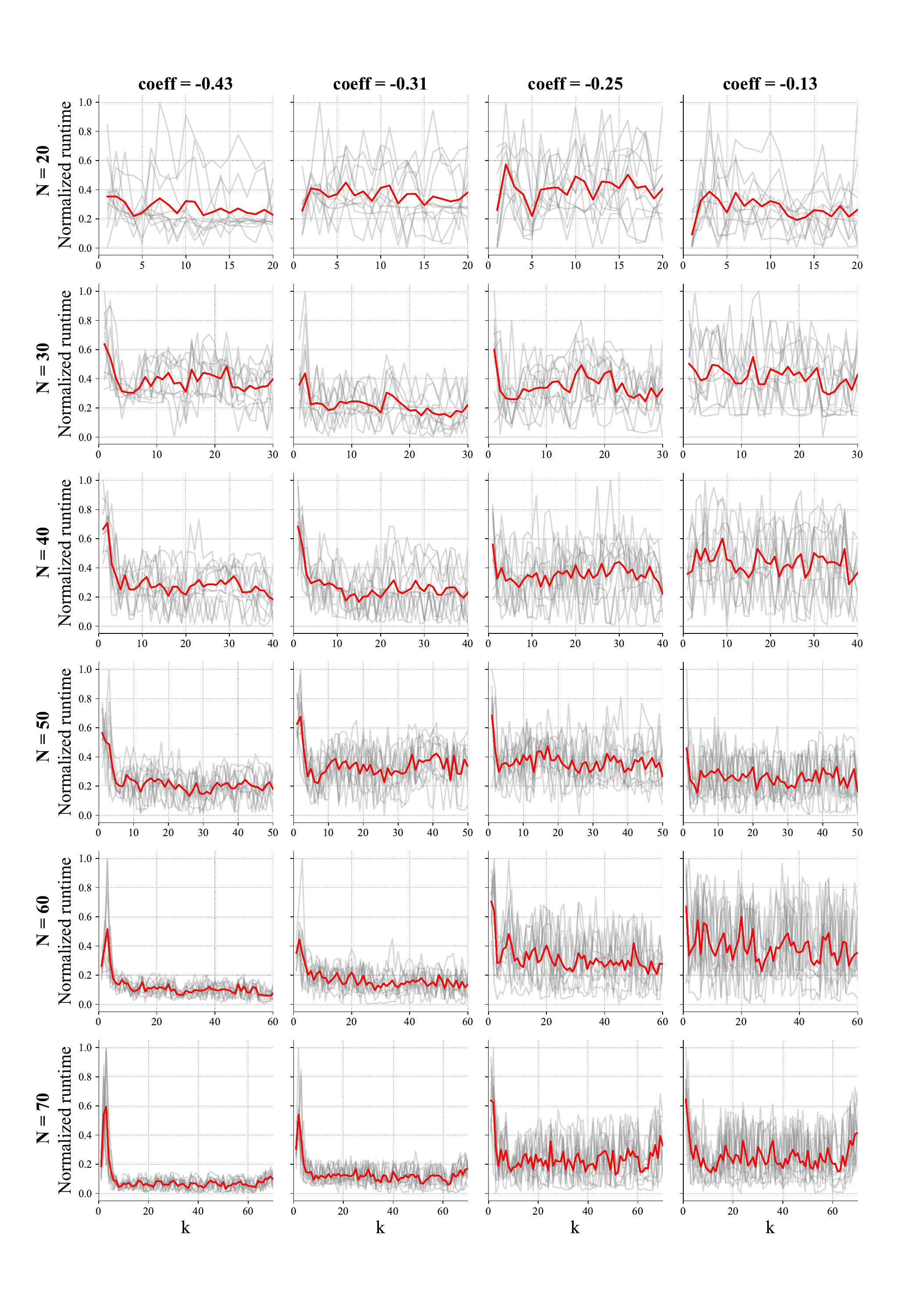}
  
\caption{Average normalized runtime across different game sizes and correlation coefficients using the single MILP.}
\label{fig:sing_milp_coeffs:runtime}
\end{figure*}

\begin{figure*}[t]
    \centering
   \includegraphics[width=1\linewidth]{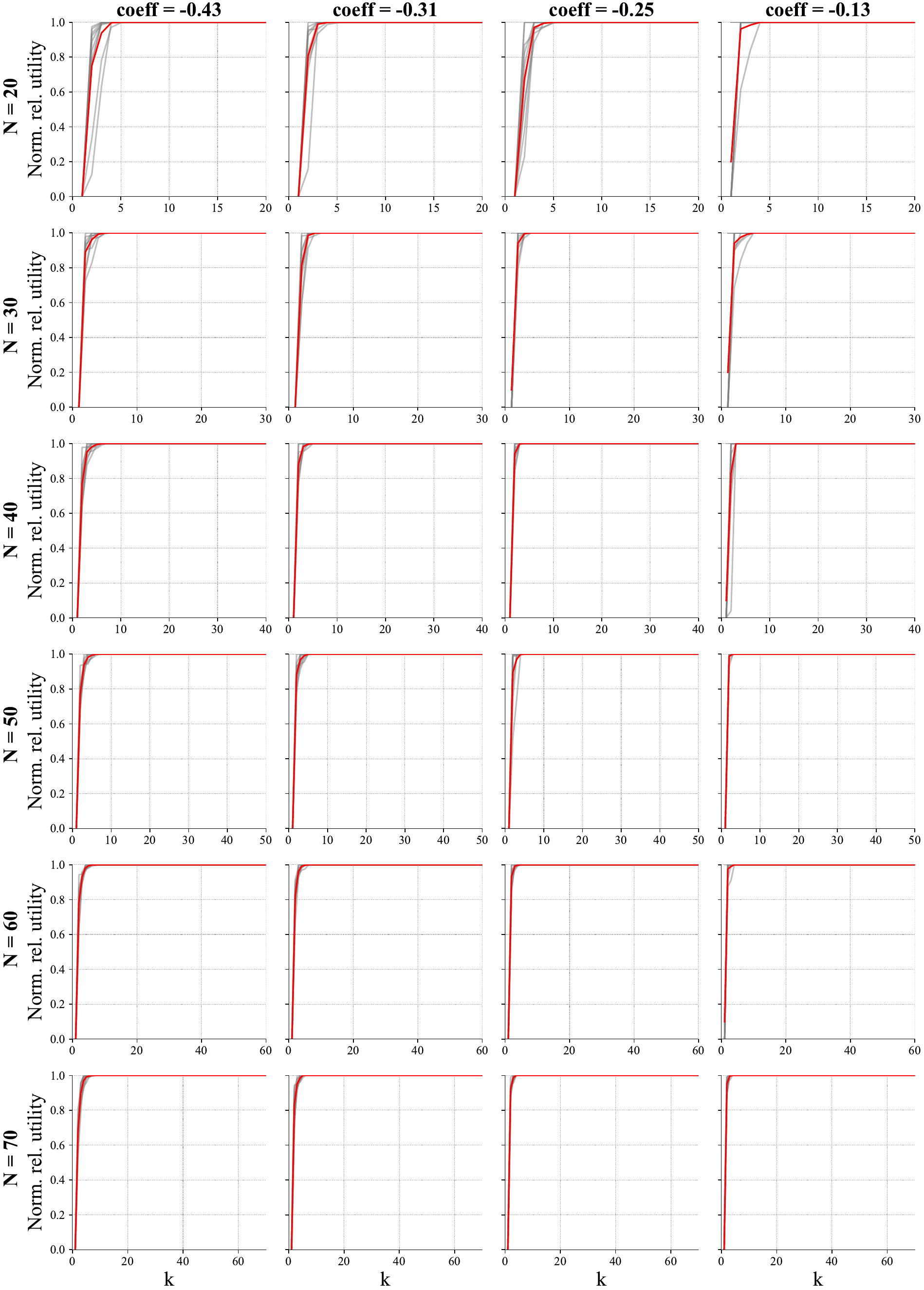}
  
\caption{Average normalized relative utility across different game sizes and correlation coefficients using the single MILP.}
\label{fig:sing_milp_coeffs:utility}
\end{figure*}

\begin{figure*}[t]
    \centering
   \includegraphics[width=1\linewidth]{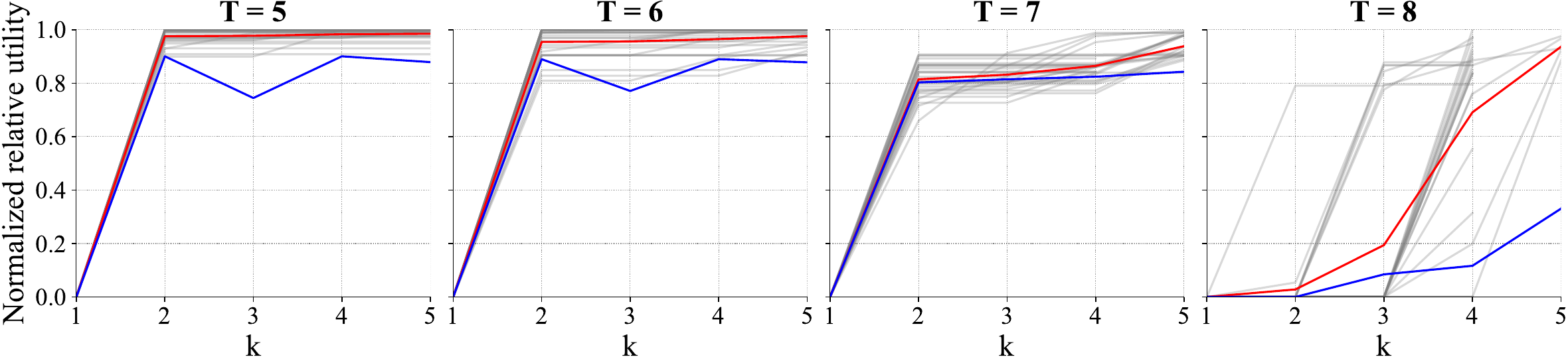}
    \includegraphics[width=1\linewidth]{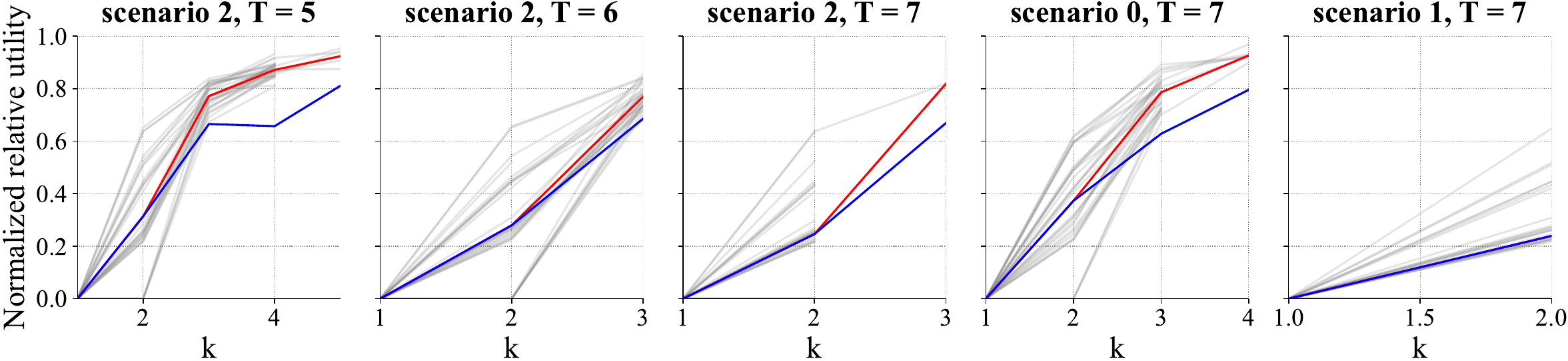}
  
\caption{Average normalized relative utility curves for optimal $k$-sparse (red) and $k$-uniform (blue) commitments. Results for individual instances using optimal $k$-sparse commitments are in gray. \textbf{Top row:} sparse player has a large strategy space (scenario 0). \textbf{Bottom row:} both players have large strategy spaces.}
\label{fig:exp2:scen0}
\end{figure*}

\begin{figure*}[t]
    \centering
   \includegraphics[width=1\linewidth]{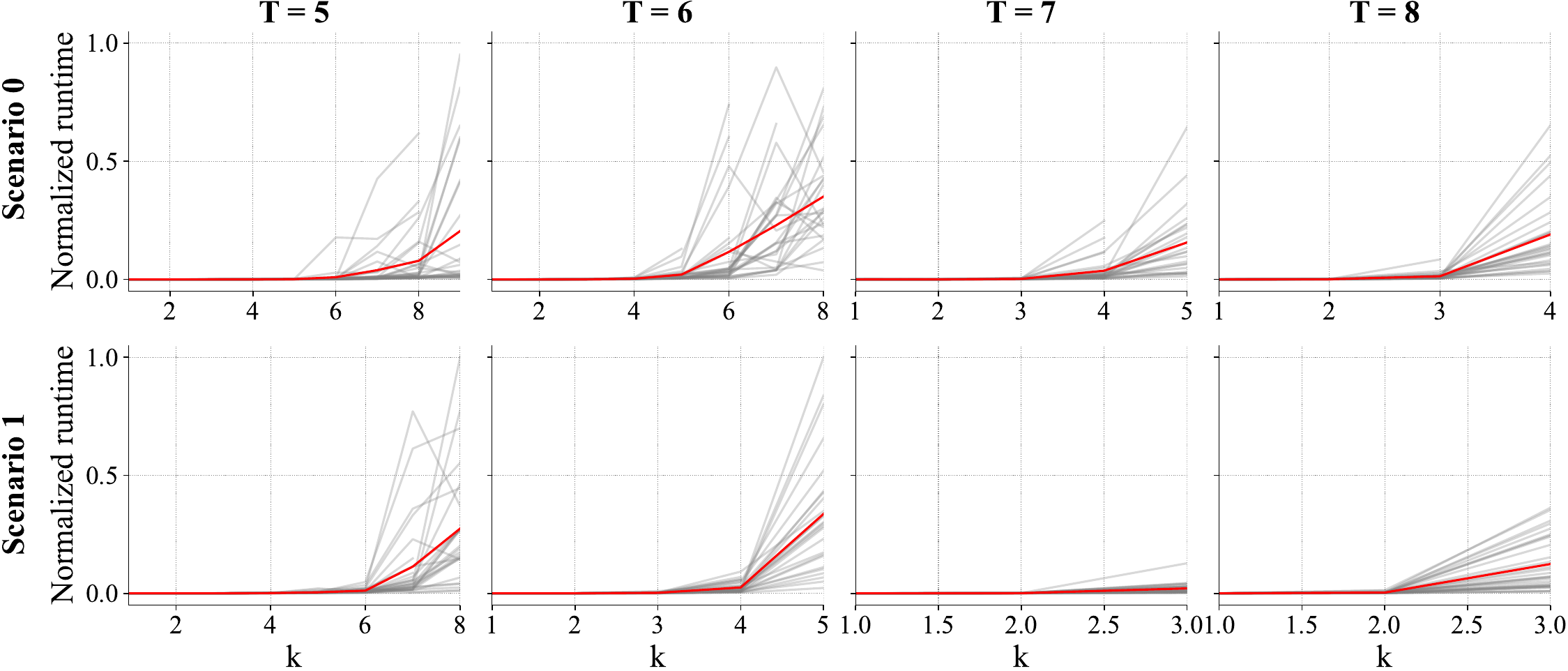}
  
\caption{Average normalized runtime (red) across 30 instances (gray) for scenarios 0 and 1 at different depth values, when the sparse player has a large strategy space.}
\label{fig:exp2:runtime}
\end{figure*}

\begin{figure*}[t]
    \centering
   \includegraphics[width=1\linewidth]{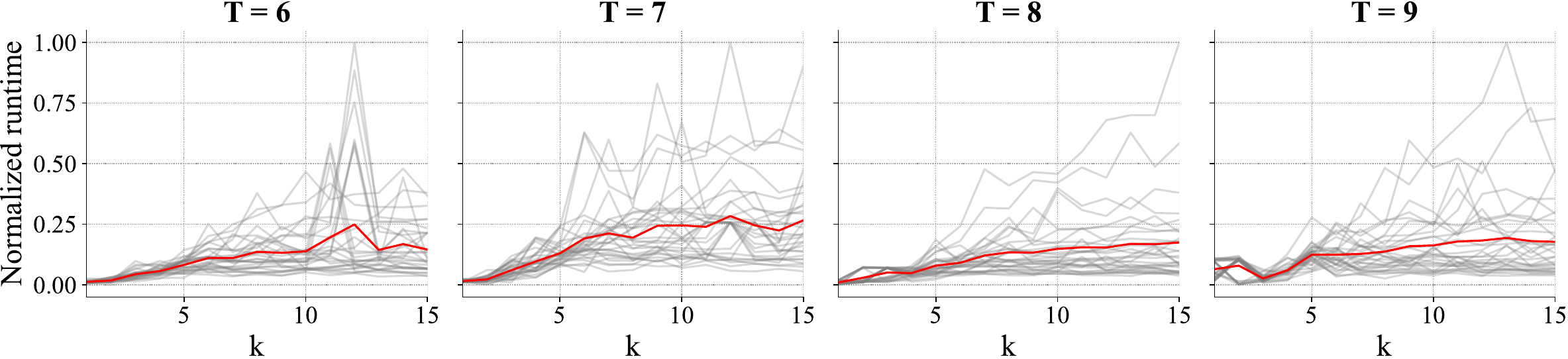}
  
\caption{Average normalized runtime (red) across 30 instances (gray) for different depth values, when the non-sparse player has a large strategy space.}
\label{fig:exp3:runtime}
\end{figure*}

\begin{figure*}[t]
    \centering
   \includegraphics[width=1\linewidth]{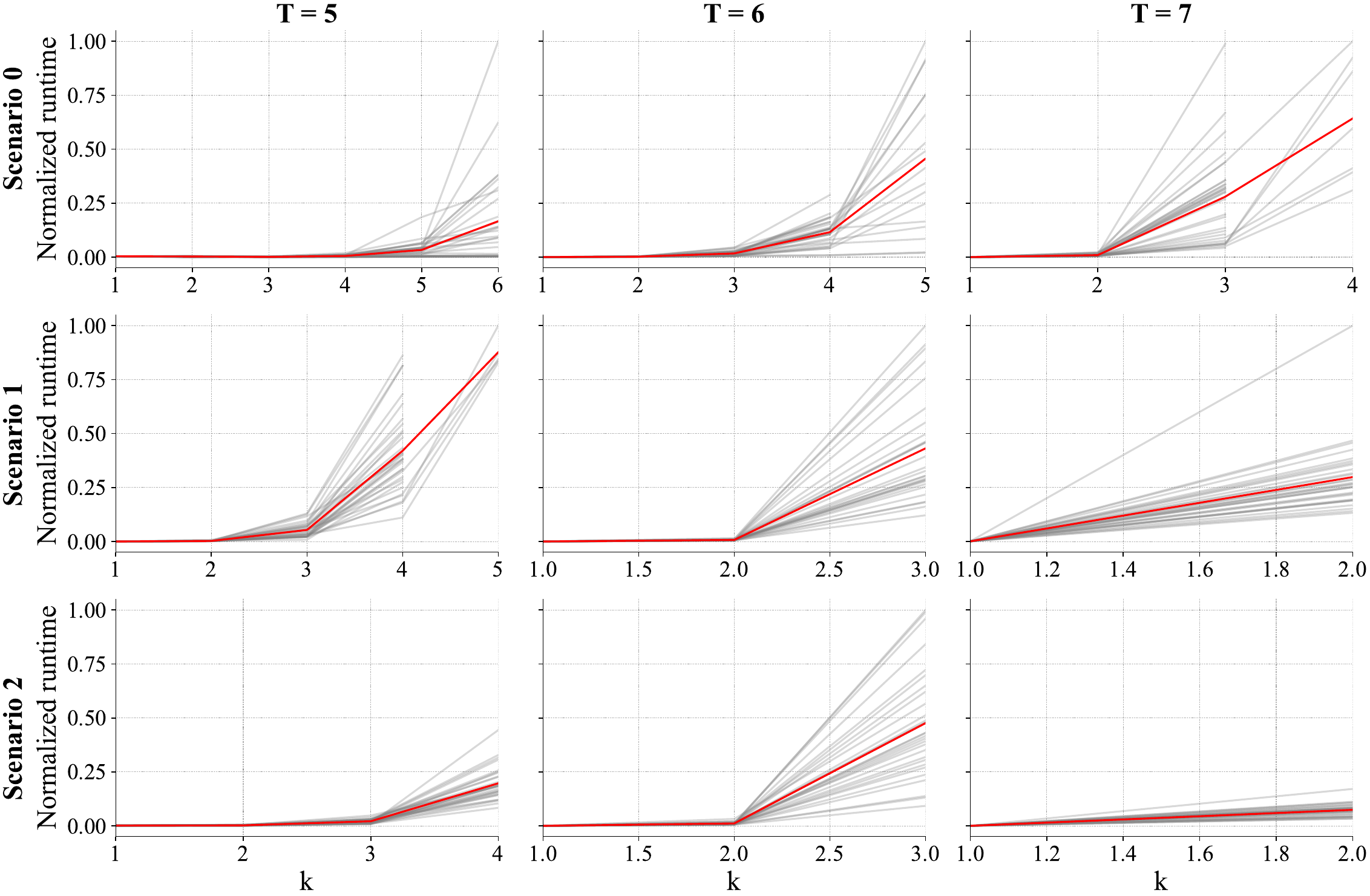}
  
\caption{Average normalized runtime (red) across 30 instances (gray) for scenarios 0, 1 and 2 at different depth values, when both players have large strategy spaces.}
\label{fig:exp4:runtime}
\end{figure*}

\end{document}